\documentclass[journal]{IEEEtran}
\IEEEoverridecommandlockouts
\usepackage{graphicx, subcaption} 
\graphicspath{ {./} }
\usepackage{amssymb}
\usepackage{latexsym}
\usepackage{soul}  % Package for striking the text
\usepackage{xcolor}  % Package for coloring the text
\newcommand{\bl}[1]{\textcolor{black}{#1}}  % Command for the text color 

\usepackage{mathtools, cuted}
\usepackage{lipsum, color}
\usepackage[font=small]{caption}
\usepackage{amsmath}
\usepackage{relsize}

\usepackage{listings}
\usepackage{textgreek}
\usepackage{amsthm}
\usepackage{amssymb}
\usepackage{mdframed}
\usepackage[bookmarks=false]{hyperref}
\theoremstyle{definition}

\usepackage{diagbox}

\usepackage{algpseudocode}
\usepackage{algorithm}
\usepackage{algorithmicx}
\usepackage{mdframed}
\usepackage{amsmath,newtxtext,newtxmath}
\algrenewcommand\algorithmicrequire{\textbf{Input:}}
\algrenewcommand\algorithmicensure{\textbf{Output:}}

\usepackage[linguistics]{forest}

\newcommand{\comment}[1]{}

\newtheorem{theorem}{Theorem}
\newtheorem{lemma}{Lemma}

\newtheorem{definition}{Definition}

\newtheorem{proposition}{Proposition}
\newtheorem{corollary}{Corollary}
\newtheorem{example}{Example}

\usepackage{mdframed}
\newtheorem{remark}{Remark}
\newmdtheoremenv{problem_stmt}{Problem}

\usepackage{cite}
\usepackage{amsmath,amssymb,amsfonts}
\usepackage{graphicx}
\usepackage{textcomp}
\usepackage{xcolor}
\def\BibTeX{{\rm B\kern-.05em{\sc i\kern-.025em b}\kern-.08em
    T\kern-.1667em\lower.7ex\hbox{E}\kern-.125emX}}

\begin{document}
\bstctlcite{IEEEexample:BSTcontrol}
% spatial-provenance Solutions for Ensuring Privacy in Wireless Communications
% "spatial-provenance Techniques for Privacy Compliance in Wireless Networks"
% Implementing spatial-provenance using Bloom filters in Privacy Constrained Networks
\title{
% On Low-Latency Strategies For Learning Spatial-Provenance with Variable-Privacy Constraints\\
On Spatial-Provenance Recovery in Wireless Networks with Relaxed-Privacy Constraints\\
% On Learning Spatial-Provenance with Variable-Privacy and Variable Coverage Constraints\\
% On Privacy-Preserving Strategies for Learning spatial-provenance in Multi-Hop Networks\\
% \thanks{This work is sponsored by the SAG, DRDO.}
% On Low Latency Strategies For Learning Spatial-Provenance with Variable-Privacy COnstraints
}

\author{\IEEEauthorblockN{Manish Bansal$^{\dagger}$, Pramsu Shrivastava$^{*}$, J. Harshan$^{\dagger, *}$,\\
$^{*}$Department of Electrical Engineering, $^{\dagger}$Bharti School of Telecommunication Technology and Management,\\
Indian Institute of Technology Delhi, India
}\\
% {Manish Bansal, J. Harshan, Pramsu Shrivastava}
 % \IEEEauthorblockA{Indian Institute of Technology Delhi, India}
}
% \and
% \IEEEauthorblockN{Pramsu Shrivastava}
% \IEEEauthorblockA{\textit{Department of Electrical Engineering} \\
% \textit{IIT Delhi, India}\\
% ee1210140@iitd.ac.in}
% \and
% \IEEEauthorblockN{Harshan Jagadeesh}
% % \IEEEauthorblockA{\textit{Department of Electrical Engineering} \\
% \textit{IIT Delhi, India}\\
% % Delhi, India \\
% jharshan@ee.iitd.ac.in}

% }
% \vspace{-10cm}
\maketitle

\begin{abstract}
In Vehicle-to-Everything (V2X) networks with multi-hop communication, Road Side Units (RSUs) intend to gather location data from the vehicles to offer various location-based services. Although vehicles use the Global Positioning System (GPS) for navigation, they may refrain from sharing their exact GPS coordinates to the RSUs due to privacy considerations. Thus, to address the localization expectations of the RSUs and the privacy concerns of the vehicles, we introduce a relaxed-privacy model wherein the vehicles share their partial location information in order to avail the location-based services. To implement this notion of relaxed-privacy, we propose a low-latency protocol for spatial-provenance recovery, wherein vehicles use correlated linear Bloom filters to embed their position information. Our proposed spatial-provenance recovery process takes into account the resolution of localization, the underlying ad hoc protocol, and the coverage range of the wireless technology used by the vehicles. Through a rigorous theoretical analysis, we present extensive analysis on the underlying trade-off between relaxed-privacy and the communication-overhead of the protocol. Finally, using a wireless testbed, we show that our proposed method requires a few bits in the packet header to provide security features such as localizing a low-power jammer executing a denial-of-service attack.\looseness=-1 
\end{abstract}


\section{Introduction}
Vehicle-to-Everything (V2X) networks are anticipated to have a crucial impact on smart mobility by improving both road safety and traffic efficiency \cite{v2v}, \cite{v2x}. These networks provide an interface for inter-vehicle communication (V2V) and vehicle-to-infrastructure (V2I) communication through the use of Road Side Units (RSUs), which act as base stations akin to cellular networks. Since mission-critical data is transmitted across these networks using an underlying wireless technology, these networks are susceptible to cyber-security threats whose threat vectors are specific to the use-case of vehicular networks \cite{v2x_attacks}. Hence, future V2X networks must have the ability to monitor the various ``observables" in the network on every packet to mitigate security risks on the vehicles \cite{vanetapp}.\looseness=-1

Among the many critical tasks, autonomous vehicles in future V2X networks will have to coordinate among each other to exchange information about location-specific roadside conditions with one another and with the RSU. Therefore, localization of the vehicles is essential for the RSU to provide Location-based Services (LBS) to the vehicles and ensure seamless operation of autonomous vehicles. From the standpoint of security, localization of vehicles is also crucial for detecting threat vectors from targeted attacks, thereby ensuring a robust and secure vehicular network. Therefore, this study proposes innovative low-latency approaches to record the location information of the vehicles in V2X networks using data-flow logs, henceforth referred to as \emph{spatial-provenance}.\looseness=-1

Localization of vehicles that are in the direct coverage range of the RSU is well studied \cite{5G_multipath,mutipath_localization_1,mutipath_localization_2, V2V_assisted_coperative_localization}. However, in practical networks, direct communication between the vehicles and the RSU may be hindered due to limitations in the transmit power of vehicles or the sparse deployment of the RSUs \cite[Section~1]{sparse_RSU}. In this situation, it is well known that multi-hop communication facilitates the reliable transmission of messages from the source vehicle to the RSU via several intermediary vehicles in an ad-hoc manner. To provide LBS in multi-hop networks, the RSU must have the capability to remotely gather information about the physical location of the vehicles that either forwarded or originated the packet. In the context of security in multi-hop networks, it is important to identify the packet forwarders \cite{nodeembedding}, determine the order in which packets are forwarded \cite{amogh}, \cite{suraj}, and determine the physical location of the vehicles to localize the origin of security threats on the forwarding vehicles.\looseness=-1
\begin{figure}[ht!]
     \centering
    % \flushleft
    \includegraphics[trim={0 0 0 0},clip,scale=0.45]{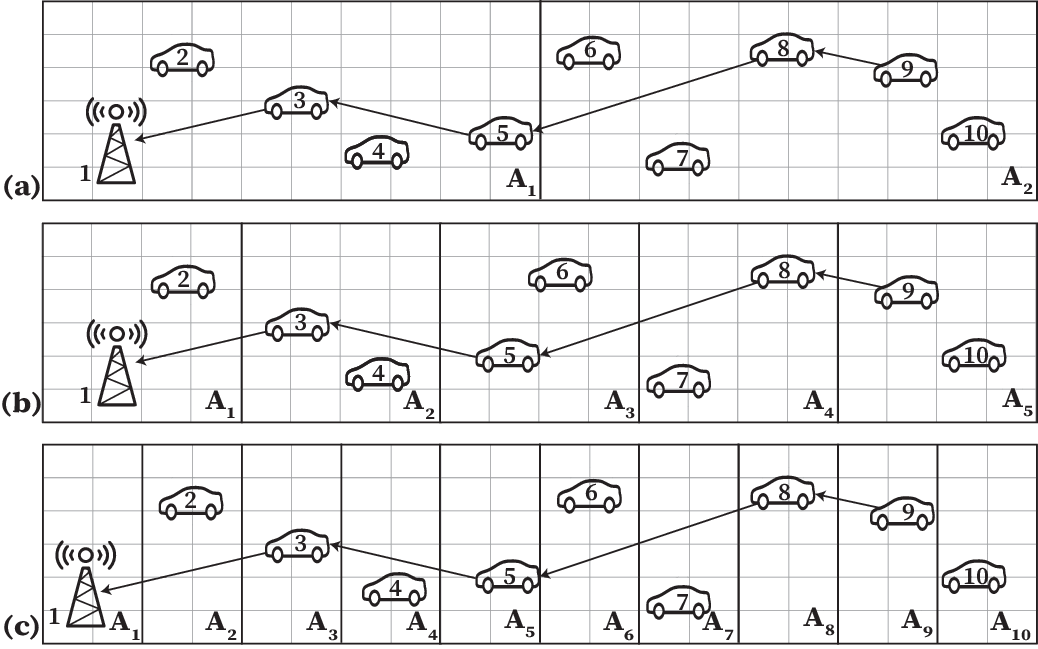}  
        \caption{A depiction of relaxed-privacy \textit{vis-à-vis} area segmentation wherein the vehicles are asked to reveal the identity of their segments instead of their exact location. Background grid depicts the precision of GPS.\looseness=-1}
        \label{fig: network Image_car}
         \vspace{-0.35cm}
\end{figure}
\subsection{Problem Statement}
Typically, vehicles utilize the Global Positioning System (GPS) for the purpose of navigation, and they have the capability to embed the same when transmitting the packet. Nevertheless, in a multi-hop communication configuration, the vehicles involved in forwarding the packet may have reservations about disclosing their precise location due to privacy concerns \cite{LPPS}. Although providing confidentiality to GPS coordinates through crypto-primitives in the packet can provide privacy to the vehicle's location from third-party observers, it also leads to an increase in end-to-end delay for the packets \cite[Section~1]{LPPS}. Moreover, it reveals the exact location of the vehicles to the RSU, which may not be desirable by the vehicles. For instance, vehicles may not want their traffic-behavioural pattern to be captured by the RSU, which can be evidently recovered by using the GPS coordinates \cite{LPPS}. However, to receive the benefits of LBS from the RSU, vehicles may be willing to reveal their location up to a certain granularity and only to the RSU. Once the granularity of localization is mutually agreed between the vehicles and the RSU, this information can be communicated to the RSU in plaintext by appending it to the packet at every forwarding vehicle. However, this approach will cause inter-vehicle privacy concerns and also increase packet size linearly with the number of hops. On the other hand, using crypto-primitives on the coarse-level location information will additionally increase the end-to-end delay of the packets. Therefore, standard embedding techniques for sharing location information have limitations in either handling privacy and growing packet size, or reducing end-to-end delay.   
To fill the above research gaps, we ask the following questions: \emph{How can we design a low-latency protocol for multi-hop wireless networks that captures the locations of the forwarding nodes while preserving the privacy of the participating nodes? Is there any analytical tool that can help system designers choose optimal parameters for the proposed protocol for different network conditions? How to validate the protocol's practicality for the V2X networks in terms of its deployment challenges?}\looseness=-1

% Is there any analytical tool that can help system designers choose optimal parameters for the proposed protocol for different network conditions?
% How to help the system designers choose optimal parameters for this new protocol?
% How to validate the protocol's practicality for the V2X networks use case in terms of its implementation aspects?

% \bl{Next, we have brief conceptual overviews of probabilistic data structures, particularly Bloom filters, to introduce the readers to the general concepts with which they may not have prior familiarity.}
% \bl{Next, we will give a brief conceptual overview of key terms.}
%\begin{comment}
\begin{table*}
\caption{\label{table:summary}We have summarized our contributions and compared them with the relevant works.}
\centering
\resizebox{16cm}{!}{
\begin{tabular}{|c|c|c|c|}
\hline
\textbf{Main contributions} & \textbf{Specific} & \textbf{Salient features} & \textbf{Existing contributions}\\
 & \textbf{contributions} & &\\
\hline
Spatial-provenance recovery & Localization in multi-hop& Does not require complex map fusion. & Complex map fusion is required in \cite{mutipath_localization_1,mutipath_localization_2}.\\ 
% \cite{5G_multipath}  uses multipath signals to estimate the location.
% with relaxed-privacy constraints  & &  & \\
with relaxed-privacy constraints && Localization of all the participating vehicles in a single packet. & \cite{V2V_assisted_coperative_localization} localizes a single vehicle using multiple vehicles. \\
% &&  & with the help of multiple vehicles. \\

&& Less number of packets required for localization. & Multiple packets are required for localization \cite{V2V_assisted_coperative_localization}.\\

&& Does not require multipath signals for localization. & \cite{5G_multipath, mutipath_localization_1,mutipath_localization_2, V2V_assisted_coperative_localization} complex calculation of multipath signals. \\
 & Communication overhead & Highly reliable, low-overhead localization in multi-hop networks.  &\cite{privacy_vanet,Location_privacy_k_anonymity, Privacy_LBS_IoV} not applicable to multihop network.\\
 % & & & high communication-overhead\\
% &&localization in multi-hop networks. &\\
 % &&&\\
  % & Relaxed-privacy & Localization of nodes with relaxed-privacy in multihop. & \cite{privacy_vanet,Location_privacy_k_anonymity, Privacy_LBS_IoV} not applicable to multihop network.\\

  &  Relaxed-privacy & Computationally efficient. & \cite{Location_privacy_k_anonymity} computationally expensive. \\
% &&  & \cite{Privacy_LBS_IoV} protects identity, but shares location in plaintext.\\
  
  % & & & and location coordinates are shared in plain text. \\
  & Variable coverage & Our method supports variable transmission ranges. & Prior works address only one wireless protocol.\\
 & Security & No changes to protocols or infrastructure are needed.   &  \\
 % \cite{freq_hop_1,Freq_hop_2,Freq_hop_3} need a change in communication protocols.

  && Localizes jammer without complex signal-strength calculations. & \cite{Jammer_localisation_1} relies on complex signal-strength calculations. \\
% && & \\
  &&& \cite{Jammer_Dos} uses Turing machines to detect jammer-immune channels.  \\
  % &&& immune to jammer.\\
  && Does not require complex ML algorithms. & \cite{CR_Vanet} uses  ML algorithms to detect rogue nodes.\\

& Spatial-provenance  & We use CLBF to obtain edge information and node locations. & \cite{amogh}, \cite{suraj} use Bloom filters to recover edge information.\\
% & recovery & and the spatial location of nodes. & recover edge information. \\

& Testbed implementation  & We implemented spatial provenance over protocols like ZigBee and LoRa.  &  \\
% &verification of spatial-provenance & over state-of-the-art protocols such as ZigBee and LoRa. &  \\
\hline
\end{tabular}}
\end{table*}
\subsection{Contributions}
To achieve a balanced integration of the RSU's localization requirements and the privacy concerns of vehicles, we assume that all the vehicles agree to disclose the same granularity\footnote{Vehicles may have varying privacy concerns, which may influence their willingness to disclose different levels of accuracy about their localization. Although privacy concerns of vehicles are generally heterogeneous, we have first focused on the homogeneous privacy of vehicles. Please refer to Section \ref{sec:discussion} for a discussion on heterogeneous privacy of the nodes.} of their location with the RSU. Furthermore, 
we suggest a \textit{divide and conquer} solution, wherein, the RSU will first divide its coverage area into equal-sized segments\footnote{\label{linear}Our idea of segmentation is applicable to intricate road models \cite{IoV_models} such as those with bi-directional roads with multiple lanes or multi roads stacked onto each other such as elevated roads, bridges on roads etc. For the latter generalized models, our ideas can be implemented either by having a dedicated RSU to each road or by having a single RSU serving multiple roads using multiple-access techniques. Please refer to Section \ref{sec:discussion} where we discuss how our proposed methodologies can be generalized to non-linear road models.} and then direct the vehicles to include the identities of their respective segments when transmitting the packets. To support this, vehicles can collectively ascertain a suitable dimension of the segment without revealing their exact positions within the segment. This choice of segment size becomes a critical factor in establishing an amicable solution to the trade-off between RSU localization requirements and vehicle location privacy. Since the exact locations of the vehicles are not shared in this model, we refer to this privacy model as \emph{relaxed-privacy} model in vehicular networks\footnote{The scope of our work extends to any wireless network in which the service provider or access point seeks to acquire knowledge about the geographical positions of the nodes inside a multi-hop network.}. For instance, in Fig. \ref{fig: network Image_car}, several possibilities are shown where the coverage area of the RSU is divided into different possible segments based on the privacy requirements of the vehicles. In the context of the stated problem, our specific contributions are summarized as follows:\looseness=-1

    1) To provide the relaxed-privacy feature, we model the roads linearly\footref{linear} in the context of vehicular networks, wherein a linear stretch of road in the communication range of the RSU is divided into linear segments of equal length, as exemplified in Fig. \ref{fig: network Image_car}. Subsequently, we propose the RSU to broadcast the segmentation information to the vehicles as a dictionary, using which, the vehicles can learn and embed their segment's identity in the packet with the help of their GPS coordinates \cite{spatial}. This way, the vehicles and the RSU can reach an amicable solution for the relaxed-privacy problem (\bl{please refer to Section \ref{sec:network_model}}).\looseness=-1
    
    2) To implement the above-proposed idea in the multi-hop setup, we propose that the vehicles use correlated linear Bloom filters (CLBF) to embed their spatial-provenance information. We highlight that CLBF is well-suited for low-latency demands in V2X networks since it offers the property of constant time for both the embedding and the recovery processes. As a salient feature of our protocol, the recovery mechanism of spatial-provenance from CLBF takes into account the underlying multi-hop protocol and the maximum transmission power of the underlying wireless technology used in the vehicular network (\bl{please refer to Section \ref{sec:Bloom Filters}}).\looseness=-1 
    
    3)  Although Bloom filters assist in low-latency recovery, it is well known that they are accompanied by error rates due to false positive events, which are dependent on Bloom filter size and the number of Hash functions used. To help the vehicles choose the appropriate parameters of the Bloom filter when using our scheme, we derive analytical expressions for error rates in recovering the spatial-provenance from CLBF. Specifically, we provide an algorithmic approach to calculate the error rates when the underlying wireless technology used in the vehicular network has an arbitrary transmission range. As a specific solution, we also provide a closed-form expression on the error rates when the wireless technology used by all the vehicles has a low range relative to the relaxed-privacy (\bl{please refer to Section \ref{sec:optimization} and Section \ref{sec:optimization_of_BF2}}). 
    To validate the impact of our method, we compare the derived analytical expressions with the simulation results, which show that the analytical expressions offer near-optimal Bloom filter parameters (\bl{please refer to Section \ref{sec:results}}).\looseness=-1
    
    4) After presenting a thorough analysis on the trade-off between privacy and error-rates, we showcase the implementation of our protocol on a testbed involving ZigBee and LoRa devices in the real world. 
    As a takeaway point, we show our technique requires a few bits in the packet header of the packet in order to help the RSU to localize the vehicles. We highlight that by virtue of the underlying tools, our protocol is applicable to large-scale V2X networks comprising an arbitrary number of vehicles in the coverage area of the RSU. We also discuss the resilience of our proposed protocols against the packet drop attacks and man-in-the-middle attacks (\bl{please refer to Section \ref{sec:Security_analysis}}).\looseness=-1 

We finally conclude by providing a security use-case of our method in vehicular networks that shows that our method can help localize low-power jammers.\looseness=-1 

\subsection{Related Works \& Novelty}
For the purpose of location privacy, various LBS solutions have been proposed in \cite{privacy_vanet, Location_privacy_k_anonymity, Privacy_LBS_IoV}. The solution in \cite{privacy_vanet} is applicable for direct communication between the RSU and the vehicle but not for multi-hop networks. Privacy-preserving localization methods in \cite{Location_privacy_k_anonymity} preserve a vehicle's location by dynamically generating virtual locations; however, they are computationally expensive. Privacy-preserving LBS in \cite{Privacy_LBS_IoV} uses a dedicated server, allocating pseudo identity, but shares GPS coordinates in plain text, making it susceptible to security threats.\looseness=-1

To help accurate localization, radio-based techniques improve precision by using existing network infrastructure without additional sensors \cite{5G_multipath}.  For such multipath-assisted localization methods, we refer to \cite{mutipath_localization_1}, \cite{mutipath_localization_2}, \cite{V2V_assisted_coperative_localization}. However, they are not applicable to multi-hop networks.\looseness = -1

V2X networks are vulnerable to attacks from malevolent adversaries due to shared wireless channels \cite{survey_wireless_security,jammer_localisation_survey, Dos_jammer}.
As a result, various anti-jamming solutions have been proposed in \cite{freq_hop_1,Freq_hop_2,routing_1} to mitigate the effects of a jamming attack and to restore regular network operations. Methods for detecting and precisely locating jammers using signal strength are proposed in \cite{Jammer_localisation_1}. However, this strategy necessitates the detection of the jamming signal intensity locally at the vehicles. Authors in \cite{antenna_identification_jammer_localisation} propose an approach for identifying and locating the jammer using the topology information of the jamming scenario, specifically for directional antennas. Authors in \cite{Jammer_Dos} classify the jammers as partial or full knowledge jammers and use turing machines to algorithmically characterize the channels that are immune to jamming attacks. To improve network security, machine learning (ML) methods that classify the rogue nodes from the primary nodes are proposed in \cite{CR_Vanet}. However, large amounts of data are needed to utilize these algorithms.\looseness=-1  

Overall, the following points define the novelty of our work.\looseness=-1 

1) \textit{Localization in multi-hop network}- Most of the earlier works have focused on the localization of the penultimate node, where the RSU gets the location of the node that is in the last hop. Our work focuses on getting the location of all the intermediate nodes of multi-hop communication within a single packet.\looseness=-1
    
2) \textit{Relaxed-privacy}- Previous works have not looked at relaxed-privacy; we tackle the relaxed-privacy problem by introducing the notion of dividing the coverage area of RSU into a number of segments.\looseness=-1 
    
3) \textit{Variable coverage}- Existing works do not look at wide range of standards with variable transmission ranges. Our method is applicable for a wide range of standards such as Wi-Fi, ZigBee, LoRa, etc.\looseness=-1
    
4) \textit{Low-latency strategy}- Because of the tools and the design methodology used, our work is suitable for the low-latency demands of V2X networks.\looseness=-1

Some preliminary versions of this work are available in \cite{manish_WCNC} and \cite{manish_COMSNETS}. In particular, \cite{manish_WCNC} addresses a solution to the relaxed-privacy problem only when the underlying vehicles have specific transmission power constraints. In contrast, this work has results applicable to arbitrary transmission power constraints. \cite{manish_WCNC} does not contain any comparison with relevant baselines. In contrast, this work has a thorough comparison with GPS-based baselines by quantifying the notion of relaxed-privacy. The contributions in \cite{manish_COMSNETS} have testbed results for studying the communication-overheads of the proposed protocol; however, this work covers testbed implementation by additionally capturing a security use-case of localizing a low-power jammer. 
%\bl{The following section provides a conceptual overview of the key terms necessary for understanding our framework.}\looseness=-1

\subsection{Key Terms}
The three important keywords relevant to our work are Bloom filter, relaxed-privacy, and spatial-provenance.\looseness=-1 

{\bf  Bloom filter:} A Bloom filter is one of the probabilistic data structures that is used for testing the membership of an element within a given set. It has several advantageous properties, such as fixed size, guarantee of no false negatives, and constant time-complexity for both insertion and query operations. A traditional Bloom filter is a binary array of size $m$ bits, where elements are inserted by setting the bits in the binary array, and the position of bits is generated using multiple Hash functions. There are many varieties of Bloom filters available in the literature \cite{BF_survey, BF_survey2}.\looseness=-1

{\bf Relaxed-privacy:} A privacy model is referred to as a relaxed-privacy model wherein the network nodes agree to disclose partial information about their attributes, instead of providing complete anonymity, in order to support the system functionality. For example, in V2X networks under relaxed-privacy model and when the attribute of interest is location of the vehicle, a vehicle may share its coarse-level location information with the RSU rather than sharing its precise GPS coordinates to protect its location privacy while still supporting network activities like routing or provenance tracking, and in return, benefiting from the location-based services provided by the RSU.\looseness=-1

{\bf Spatial-provenance:} Spatial-provenance refers to the information on the identity of the forwarding nodes along with their location in a multihop network when forwarding the data packet from the source node towards the destination node. In the context of V2X networks, spatial-provenance is a generalized definition of provenance that is already known in literature \cite{nodeembedding, amogh}. It is beneficial for authenticating the data origin, maintaining trustworthiness in movement-aware applications, and help in detecting abnormal routing behavior.\looseness=-1
\vspace{-0.2cm}
\subsection{Executive Summary}\label{sec:Executive_summary}
Section \ref{sec:network_model} presents the foundational V2X use case that underpins the design of the proposed protocol, which outlines the underlying assumptions related to the road model, routing and communication constraints, and the mobility of vehicles with respect to the RSU. Section \ref{sec:Bloom Filters} provides the core mechanisms of the proposed protocol, where a basic understanding of the Bloom filters helps readers to grasp the working principles of our proposed protocol. Section \ref{sec:optimization_of_BF2} helps system designers to choose optimal parameters for implementing the proposed protocol for varying network conditions, where we present a mathematical framework and a low-complexity algorithm for optimizing the key parameters involved in embedding the spatial-provenance in the CLBF. Section \ref{sec:results} provides experimental results to showcase the efficacy of the mathematical framework derived in Section \ref{sec:optimization_of_BF2} and demonstrate the practical effectiveness of our approach by comparing it with GPS-based localization methods. Section \ref{sec:testbed_results_practical_aspects} provides insights into the practical implementation of our proposed protocol on real-world V2X networks, such as compression of provenance, testbed results, and other details related to the mobility patterns of the vehicles. Section \ref{sec:Security_analysis} explores the privacy and security dimensions of our proposed protocol. In Section \ref{sec:discussion}, we conclude with a discussion about how our suggested provenance protocols can be incorporated into more complex road network scenarios.\looseness=-1

\section{Network Model for Joint Localization and Provenance} \label{sec:network_model}

We consider a set of $N$ nodes, out of which $N-1$ nodes are mobile vehicles, and one node is the stationary RSU. This network model can be seen as a snapshot of a vehicular network wherein the mobiles vehicles are in the coverage area of an RSU by the virtue of their mobility.\footnote{For authenticating a vehicle in its coverage area, we may assume that gateway authentication techniques \cite[Section~2.2]{suraj} are applied by the RSU.} Mobile nodes want to communicate with the RSU; however, due to limited power constraints, they cannot communicate directly with the RSU.\footnote{Communication messages in the context of V2X networks can be categorized into two main types, namely: location-based service messages and cooperative awareness messages \cite{LPPS}. The communication messages in this context mostly belong to the category of LBS class.} Therefore, nodes communicate with the RSU with the help of several intermediate nodes in a multi-hop manner using a decode and forward strategy. To enhance the network's security and to offer advanced network-diagnostic capabilities, the RSU aims to acquire knowledge of the path traveled by the packets, and the precise locations of all the nodes in the network that have relayed the packets.
However, due to privacy concerns, vehicles may not share their exact GPS coordinates with the RSU \cite[Section~1]{LPPS}. Therefore, to balance the requirements of the RSU and the privacy of vehicles, the RSU divides its coverage area into $r$ geographical regions, for some $r \in \mathbb{N}$, and requests the nodes to share the identities of their regions while they forward the packet. In this context, $r$ can be an arbitrary number, which is \emph{a priori} fixed based on an amicable solution between the privacy requirement of the nodes and the precision requirement of the RSU.\looseness=-1

Given the use-case of vehicular networks, we assume that the RSU models its coverage area along a straight line and divides it into smaller regions of equal length, henceforth referred to as segments. The set of segments is denoted by $\Delta \triangleq \{A_1,A_2,\ldots,A_{r}\}$, as exemplified in Fig. \ref{fig: network Image_car} for $r =2,5$ or $10$. We assume that the RSU is in area $A_1$, i.e., at one corner of the coverage area, whereas the other mobile nodes are randomly distributed across $\Delta$. We also assume that the segments farther from the RSU get higher index numbers\footnote{For exposition, the positioning of the RSU is at one end of the coverage area, i.e., in $A_1$. However, in practice, given the omnidirectional nature of radio propagation, the RSU can also cater to vehicles on the other side of the road with identical segmentation via orthogonal radio resources.}. For instance, the farthest segment in the coverage area is $A_{r}$. The above-mentioned segmentation of the coverage area is designed to ensure that the nodes learn their segment identity and embed the same when forwarding the packet to the RSU. This way, the relaxed-privacy constraint of the nodes is preserved with the granularity captured by $\Delta$. To assist the nodes in learning their segments, the RSU broadcasts a dictionary containing the boundaries of GPS locations mapped to different segments\footnote{ The broadcast of the dictionary can either be continuous or periodic. However, to avoid excessive channel usage and unnecessary energy consumption for both the RSU and the vehicles, we suggest that the RSU broadcasts the dictionary periodically, as long as the broadcast interval is appropriately chosen. For more details on the dictionary broadcast, please refer to Section \ref{sec:Dictionary}}. Since the nodes are equipped with GPS, they can privately derive their segment identity (ID) using the broadcasted dictionary. For details on sharing the dictionary privately through a broadcast message, we refer the reader to \cite{spatial}. It is assumed that the RSU is a high-power device; therefore, the broadcast messages of RSU can reach the farthest segment, i.e., $A_r$ on the downlink.\looseness=-1  

To elucidate the spatial-provenance embedding procedure, the identity of the source node is designated as $I_s$, the identity of the RSU is $I_{N}$, and the identities of the remaining mobile nodes are denoted as $I_n$, where $n \in [1,~N-2]$. Additionally, to capture the mapping between $\Delta$ and the nodes' GPS coordinates, we employ the function $g: \{I_1, I_2,\ldots,I_{N-2}, I_s\} \rightarrow \Delta $ for translating the node identities into segment identities. We assume that the function $g$ is a constant mapping for a given coherent time when the packet is en-route from the source node $I_s$ towards the RSU. This assumption follows from the fact that the time taken by the vehicles to cross consecutive segments is typically much smaller compared to the end-to-end delay on the packet \cite[Section~2]{suraj}. However, the function $g$ may change at a later point in time when another packet is being transferred. The following section describes the routing protocol adopted by the nodes for their uplink communication with the RSU.\looseness=-1
\begin{figure}
    \centering
    \includegraphics[scale=1.2]{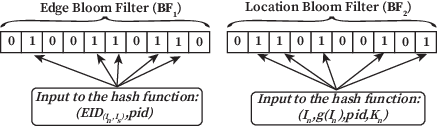}
    \caption{The diagram illustrates the embedding process in CLBF at the intermediate nodes, which consist of edge and location Bloom filters with sizes of $m_1 = 10$ bits and $m_2 = 10$ bits, respectively. We employ $k_1=k_2=5$ Hash functions to embed an element in CLBF.\looseness=-1
    % Depiction of embedding in CLBF at the intermediate nodes which comprise edge and location Bloom filters of size $m_1 =8$ bits and $m_2=8$ bits, respectively. Here $k_1=k_2=3$ number of Hash functions are used for embedding an element in CLBF.\looseness=-1
    }
    \label{fig:CLBF}
    % \vsapce{-0.3cm}
    \vspace{-0.4cm}
\end{figure}
\vspace{-0.4cm}
\subsection{Routing Constraints} \label{sec:Routing Constraints}
We assume that the following routing and communication constraints are followed in the network when the packets are communicated from the source node $I_{s}$ to the RSU.\looseness=-1

1) Routing Constraint: A routing protocol such as Ad-hoc on-demand Distance Vector (AODV) \cite{aodv_original} is used in the network, which ensures minimum-hop delivery of the packet from the source to the RSU. Therefore, the underlying routing protocol prevents the formation of any loop or back-hop. Formally, a node in the segment $A_i$, for $i \in [r]$ cannot send the packet to a node in the segment $A_j$, where $j \in [r]$ such that $i< j$.\looseness=-1

2) Communication Constraint: We assume that nodes have limited power constraints and are homogeneous in terms of their transmission range, i.e., each node can communicate only up to a certain distance. To capture this communication constraint, we use the parameter $\beta$, for $\beta \in \mathbb{N}$, which denotes the maximum number of adjacent segments a mobile node can communicate across. For example, if $\beta=2$ and a node is in segment $A_3$, it can only send the packets to the nodes in segment $A_3$, or $A_2$, or $A_1$, and likewise, can receive the packets from the nodes in segments $A_3$, $A_4$, or $A_5$. To mathematically capture this communication constraint, we use the function $S: \Delta \rightarrow [r]$, which translates the segment identity $A_i$ into their numerical identity $i$, for $i \in [r]$. Furthermore, with $I_i$ and $I_j$ representing the identities of two consecutive nodes in a $h$-hop path, they must satisfy the communication constraint $ |S(g(I_i))-S(g(I_j))|\leq \beta$, where $i\neq j$ and $i, j \in [N]\backslash \{N-1\} \cup \{s\}$.\looseness=-1 

With the above routing constraints, the objective of the nodes is to embed the information on their segment identities, as well as the path traveled by the packets, so that the same can be recovered at the RSU upon receiving the packets. In the context of latency-constrained V2X networks, it is typically desired to maintain a constant packet size during multi-hop communication. Therefore, in the next section, we present a Bloom filter \cite{Bloom_filter_1} based data-structure to achieve the above-mentioned objectives.\looseness=-1
\vspace{-0.3cm}
\section{Bloom Filter based Strategy for Recovering Spatial-Provenance} \label{sec:Bloom Filters}
To retrieve the path traveled by the packet and segment IDs of the nodes that have relayed it, we employ Correlated Linear Bloom Filters (CLBF) that extend the capabilities of the standard Bloom filter by enabling concurrent insertions and queries from multiple distinct sets to multiple Bloom filters. In the context of this work, we use CLBF with two Bloom filters, as shown in Fig. \ref{fig:CLBF}, wherein one will be used for embedding the information on the path traversed by the packet, and the other for carrying the information on the spatial location of the nodes. The choice of utilizing a Bloom filter for this purpose is driven by several advantageous properties, such as fixed size, guarantee of no false negatives, and constant time-complexity for both insertion and query operations.\looseness=-1
\vspace{-0.25cm}
\subsection{Ingredients for Embedding Spatial-Provenance} \label{sec:BF_ingredients}
To facilitate the provenance embedding process, the packet structure comprises a CLBF along with a counter to carry the hop-length information to the RSU. The CLBF has two Bloom filters, denoted as $\textbf{BF}_{1} \in \{0, 1\}^{m_{1}}$ and $\textbf{BF}_{2} \in \{0, 1\}^{m_{2}}$, wherein $\textbf{BF}_{1}$ is intended for embedding information on the path traveled by the packet, and $\textbf{BF}_{2}$ is for embedding information on the segment IDs of the nodes. At a high-level, when each node forwards the packet, it performs three tasks, namely (i) it embeds the identity of the edge connecting itself and the preceding node onto $\textbf{BF}_1$ using $k_{1}$ number of Hash functions for $1\leq k_1\leq m_1$ \cite{amogh, suraj}, (ii) it embeds the identity of its segment onto $\textbf{BF}_2$ using $k_{2}$ number of Hash functions for $1\leq k_2 \leq m_2$, and (iii) it increments the hop-counter. As a result, upon receiving the packet, the RSU will be able to jointly recover the path traveled by the packet and the corresponding location of the forwarding nodes.\looseness=-1 

For the edge embedding process onto $\textbf{BF}_{1}$, we assume that each node $I_i$, where $i \in [N] \backslash \{N-1\} \cup \{s\}$, is equipped with a pre-shared key with the RSU, denoted as $K_i$. Through mutual authentication with its neighbouring node, denoted as node $I_j$, where $j\neq i$, node $I_i$ obtains a distinct edge ID denoted as $EID_{(i,j)}$. One such approach to produce $EID_{(i,j)}$ involves utilizing $I_i$ and $I_j$ in conjunction with the key $K_i$ to derive a string $EID_{(i,j)} = Enc_{K_{i}}(I_i||I_j)$, where $Enc$ represents a suitable encryption method such as the Advanced Encryption Standard (AES). Note that the function $EID_{(\cdot, \cdot)}$, is not symmetric, i.e., $EID_{(i,j)}$ is not equal to $EID_{(j,i)}$, where $EID_{(j,i)} = Enc_{K_{j}}(I_{j} || I_{i})$ is the edge ID derived at node $I_j$ for the link between node $I_j$ and node $I_i$. We assume that the RSU has access to the IDs and secret keys of all the nodes, allowing it to obtain all the set of edge IDs in the network, represented as $\mathcal{E} = \{EID_{(i, j)}~|~ \forall ~i, j, ~j \neq i\}$.\looseness=-1

For embedding segment IDs onto $\textbf{BF}_{2}$, we assume that node $I_{i}$ has the knowledge of its segment using the dictionary broadcast by the RSU. For reference, Table \ref{tab:notation} provides a list of all the notations used in the description of the protocol. Using the above-mentioned ingredients, the following section describes our method to embed the spatial-provenance on $\textbf{BF}_{1}$ and $\textbf{BF}_{2}$.\looseness=-1

\begin{table} 
   \centering
   \caption{\label{tab:notation} Notations used in this paper.}
   \resizebox{7.8 cm}{!}{
   \begin{tabular}{|c|c|}
     \hline
     \textbf{Term} & \textbf{Meaning} \\
     \hline
     $N$ & Number of nodes in the network \\
     \hline
     % $\mathbf{BF}$ & Bloom filter array \\
     % \hline
     % $m$ & Bloom filter size (in bits) \\
     % \hline
     % $k$ & Number of the Hash function used \\
     % \hline
      $m_1$, $m_2$ & Size of $\textbf{BF}_1$ and $\textbf{BF}_2$, respectively (in bits) \\
     \hline
     $k_1$, $k_2$ & Number of the Hash function used in $\textbf{BF}_1$ and $\textbf{BF}_2$, respectively \\
     \hline
     $k_{1}^{*}$, $k_{2}^{*}$ & Optimal number of Hash function used in $\textbf{BF}_1$ and $\textbf{BF}_2$, respectively \\
     \hline
     $h$ & Number of hops traveled by the packet \\
     \hline
     % $seq$ & Sequence number of packet \\
     % \hline
    % %  $BF$ & Bloom filter  array \\
    %  \hline
     $I_i$ & ID of $i^{th}$ node \\
     \hline
     $\beta$ & Communication constraint \\
     \hline
     % $K_i$ & Shared Secret key between $n_i$ and destination \\
     % \hline
     $(i,j)$ & Edge connecting node $I_j$ to $I_i$, i.e., $I_j \rightarrow I_i$\\
     \hline
     $\Delta, ~|\Delta|$ & Set of segment IDs, cardinality of set $\Delta$, i.e., $|\Delta|=r$\\
     \hline
     $\alpha$ & Number of random bits lit in $\textbf{BF}_2$\\
     \hline
     $pid$ & Packet identity (ID)\\
     \hline
     $g ()$ & A function to translate $I_i \rightarrow \Delta$\\
     \hline
     $S()$ & A function to translate $\Delta \rightarrow [r]$\\
     \hline
     $\mathcal{C}$ & Set of $h$-hop path from $I_s$ to the RSU\\
     \hline
     $\mathcal{G}_{\beta}$ & Set of valid mappings $g(\cdot)$ between the node IDs and $\Delta$\\
     \hline
     $\mathcal{F}$ & Set of false pairs\\
     \hline
     $\mathcal{R}$ & Set of extra recoveries\\
     \hline
     % $E$ & Set of all edges in network\\
    
     % \hline
     \end{tabular}}
     \vspace{-0.3cm}
 \end{table} 
 \vspace{-0.3cm}
\subsection{Embedding Process for Spatial-Provenance} \label{sec:BF_embedding}

To explain the provenance embedding process, we assume that the packets are forwarded from the source node $I_{s}$ through a $h$-hop path $I_{s} \rightarrow I_{j_{1}} \rightarrow I_{j_{2}} \rightarrow \ldots \rightarrow I_{j_{h-2}} \rightarrow I_{j_{h-1}}$ before reaching the RSU. To start with, both Bloom filters are initialized to all-zero vectors, and the counter is also initialized to zero. At the source node $I_s$, $\textbf{BF}_1$ is untouched, whereas $\textbf{BF}_2$ is updated as follows. Using a Hash function $H_{m_2}(\cdot)$ that randomly picks an index in $[1, m_{2}]$ with uniform distribution, the source node $I_{s}$ generates $k_{2}$ indices denoted by $\{v_{I_s}^{(l)} \in [1, m_{2}]| 1\leq l \leq k_{2}\}$ in an independent manner as\looseness=-1
\begin{equation}
\label{eq:embed_BF2}
v_{I_s}^{(l)} = H_{m_2}(I_s||g(I_s)||pid||l||K_{s}),
\end{equation}
where $pid$ denotes the packet ID, and all the other notations are as explained earlier. Subsequently, the Bloom filter $\textbf{BF}_{2}$ is set to one as $\textbf{BF}_{2}[v_{I_s}^{(l)}] = 1$ on the chosen set of indices $\{v_{I_s}^{(l)} | 1\leq l \leq k_{2}\}$. Once $\textbf{BF}_2$ is updated at node $I_s$, the hop-counter is incremented by one, and the packet is forwarded to the next node $I_{j_{1}}$.\looseness=-1 

Upon receiving the packet from node $I_{s}$, node $I_{j_{1}}$ updates both $\textbf{BF}_{1}$ and $\textbf{BF}_{2}$ as follows. To embed the edge ID connecting node $I_{s}$ and $I_{j_{1}}$ onto $\textbf{BF}_{1}$, a Hash function $H_{m_1}(\cdot)$ that randomly picks an index in $[1, m_{1}]$ with uniform distribution is used. Then node $I_{j_{1}}$ generates $k_{1}$ indices denoted by $\{w_{I_{j_{1}}}^{(a)} \in [1, m_{1}]| ~1\leq l \leq k_{1}\}$ in an independent manner as\looseness=-1
\begin{equation}
\label{eq:embed_BF1}
w_{I_{j_{1}}}^{(a)} = H_{m_1}(EID_{(I_{j_{1}},I_{s})}||pid||a),
\end{equation}
where $EID_{(I_{j_{1}},I_{s})}$ denotes the edge ID for the link $I_{s} \rightarrow I_{j_{1}}$ generated at $I_{j_{1}}$, and all the other notations are as explained earlier. Subsequently, the Bloom filter $\textbf{BF}_{1}$ is set to one as $\textbf{BF}_{1}[w_{I_{j_{1}}}^{(a)}] = 1$ on the chosen set of indices $\{w_{I_{j_{1}}}^{(a)} |~ 1\leq a \leq k_{1}\}$. After embedding $\textbf{BF}_{1}$, node $I_{j_{1}}$ proceeds to embed its segment ID onto $\textbf{BF}_{2}$ in a manner similar to that explained for node $I_{s}$. Finally, the hop-counter is incremented by one before forwarding the packet to the next node, i.e., node $I_{j_{2}}$. Along similar lines, $\textbf{BF}_{1}$,  $\textbf{BF}_{2}$, and the hop-counter are updated at every node while the packet reaches the RSU through the rest of the path $I_{j_{2}} \rightarrow I_{j_{3}} \rightarrow \ldots \rightarrow I_{j_{h-2}} \rightarrow I_{j_{h-1}}$.\looseness=-1
\vspace{-0.2cm}
\subsection{Recovery Process for Spatial-Provenance} \label{sec:recovery}
Following successful packet reception at the RSU, it endeavours to trace the packet's path and identify the segment IDs of the forwarding nodes. The recovery process happens in two phases; first, the RSU recovers the path traveled by the packet from the received Bloom filter, $\textbf{BF}_1$, and then it recovers the segment IDs from the received Bloom filter $\textbf{BF}_2$.\looseness=-1 

By using $\textbf{BF}_{1}$, the RSU tests all the potential edges in the network, denoted by $\mathcal{E}$, to compile a set of recovered edges denoted as $\hat{\mathcal{E}}$, where $\hat{\mathcal{E}} \subseteq \mathcal{E}$. To recover the edges from $\textbf{BF}_1$, the RSU has all the ingredients such as node IDs $\{I_i\}$, packet ID $pid$, pair-wise keys $\{K_i\}$, edge IDs $\{EID_{(i,j)}\}$. In this context, a recovered edge from node $I_{i}$ to node $I_{j}$ satisfies the condition of $\textbf{BF}_{1}[u_{I_{j}}^{(a)}] = 1$ for all the $k_1$ indices, given by $\{u_{I_{j}}^{(a)} \in [1,m_1]~|~ 1\leq a \leq k_1\}$, which in turn are generated at the RSU using the edge IDs and all other credentials along the similar lines of \eqref{eq:embed_BF1}. Subsequently, the RSU applies a depth-first-search algorithm on the recovered set of edges $\hat{\mathcal{E}}$ to obtain a candidate set of $h$-hop paths from the source node $I_s$ to the RSU, denoted as $\mathcal{C} \triangleq \{\mathbf{c}_1, \mathbf{c}_2,\ldots,\mathbf{c}_{|\mathcal{C}|}\}$, where each $\mathbf{c}_i$ is a $h$-hop path.\looseness=-1 

From each candidate $h$-hop path $\mathbf{c}_i$ consisting of edges $\{(I_s, I_{i_1}), (I_{i_1},I_{i_2} )\ldots (I_{i_{h-1}}, I_{N})\}$, the RSU extracts the set of unique node IDs, denoted as $\mathcal{V}_i = \{I_{s}, I_{i_{1}},I_{i_{2}}, \ldots, I_{i_{h-1}} \}$.  
% edge set to obtain the correct paths of the desired hop length.
Using the set of unique node IDs $\mathcal{V}_i$, the RSU obtains the cross-product $\mathcal{Q}_i = \mathcal{V}_i \times \Delta$ to form all possible pairs of node IDs and the segment IDs. Then it checks the presence of each element of $\mathcal{Q}_i$ in $\textbf{BF}_2$ by reconstructing the set of $k_{2}$ indices along the similar lines of \eqref{eq:embed_BF2}. In this context, a pair in $\mathcal{Q}_i$ is referred to as a recovered pair if the values in $\textbf{BF}_{2}$ on all the chosen $k_{2}$ locations are set to one. Finally, for the chosen path $\mathbf{c}_{i}$, a set of potential segment IDs, which is a subset of $\Delta$, are recovered against each node, and these are denoted by $\Delta_{i} = \{\Delta_{I_s}, \Delta_{I_{i_{1}}}, \Delta_{I_{i_{2}}}, \ldots, \Delta_{I_{i_{h-1}}}\}$.\looseness=-1 

Overall, using $\textbf{BF}_{1}$ and $\textbf{BF}_{2}$, the RSU generates the set $\{(\mathbf{c_{i}}, \Delta_{i})\} ~|~ i = 1, 2, \ldots, |\mathcal{C}|]$ against the received packet with ID $pid$. Assuming that the RSU has the knowledge of the routing and communication constraints listed in Section \ref{sec:Routing Constraints}, it further prunes the list of candidate paths and their associated segment IDs from the set $\{(\mathbf{c_{i}}, \Delta_{i})\} ~|~ i = 1, 2, \ldots, |\mathcal{C}|]$ to obtain $\{(\mathbf{c_{i}}, \bar{\Delta}_{i})\} ~|~ i = 1, 2, \ldots, |\mathcal{C}|]$, where $\bar{\Delta}_{i} = \{\bar{\Delta}_{I_s}, \bar{\Delta}_{I_{i_{1}}}, \bar{\Delta}_{I_{i_{2}}}, \ldots, \bar{\Delta}_{I_{i_{h-1}}}\}$ such that $\bar{\Delta}_{I} \subseteq \Delta_{I} ~\forall \, I \in \{I_s, I_{i_1}, \ldots, I_{i_{h-1}}\}$ are obtained after applying the constraints in Section \ref{sec:Routing Constraints}.\looseness=-1

As the packet can travel only one path, ideally the RSU should recover exactly one path, and only one segment ID against every participating node, i.e., $|\mathcal{C}| =1$ and $|\bar{\Delta}_{I}| = 1, ~\forall \, I \in \{I_s, I_{j_1}, \ldots, I_{j_{h-1}}\}$. However, Bloom filters that are known to save space and time in the embedding and query process are also known to result in false positive events, especially if their parameters $m_1,~ m_2,~k_1$ and $k_2$ are not appropriately chosen. In the context of this work, the recovery process of the spatial-provenance results in a false-positive event if the RSU recovers multiple solutions to the path and segment locations satisfying the constraints in Section \ref{sec:Routing Constraints}, i.e., $|\mathcal{C}| > 1$ or $|\mathcal{C}| = 1$ with $|\bar{\Delta}_{I}| > 1,$ for some $I \in \{I_s, I_{j_1}, \ldots, I_{j_{h-1}}\}$. Therefore, in order to recover one path and only one segment ID against every participating node with a very high probability, we need to select the Bloom filter parameters $m_1,~ m_2,~k_1$ and $k_2$ appropriately by analyzing the error events related to our network model. Thus, in the next section, we present a detailed analysis of the false-positive events associated with our protocol, and then address the optimization of the Bloom filter parameters.\looseness=-1
\vspace{-0.25cm}
\subsection{Optimization of Bloom Filter Parameters}  \label{sec:optimization}
Let us consider a scenario wherein the locations of all the nodes are fixed and the packet reaches the RSU through a fixed hop-length $h$ implementing our proposed protocol. Upon applying the spatial-provenance recovery method, let $E_{fp}^{(1,2)}$ represent an event where more than one possible arrangement of nodes and segments, satisfying the constraints in Section \ref{sec:Routing Constraints}, are retrieved from the Bloom filters $\textbf{BF}_1$ and $\textbf{BF}_2$. With such events, the RSU will not be able to resolve the actual path traveled by the packet, as well as the true locations of the forwarding nodes. To ensure high reliability in the spatial-provenance recovery process, the frequency of occurrence of such false positive events, represented by $\Pr{(E_{fp}^{(1,2)})}$, should be minimized when a number of packets flow through the network. Therefore, before deploying the proposed protocol for a given hop-length, we need to choose the parameters $m_1,~ m_2,~k_1$ and $k_2$ for $\textbf{BF}_{1}$ and $\textbf{BF}_{2}$ so as to minimize the expected probability of false positive events, when averaged over all the valid spatial-distribution of the forwarding nodes that adheres to the constraints defined in Section \ref{sec:Routing Constraints}. In light of these considerations, we propose the following problem statement.\looseness=-1  
% under the assumption that the sum of the size of the Bloom filters $\textbf{BF}_1$ and $\textbf{BF}_2$ is fixed, denoted as $m$.\looseness=-1 
% Table \ref{tab:notation} lists the notations used in this paper.  Therefore, we propose the following problem statement.\looseness=-1

{\begin{problem_stmt}{}
\label{prob1}
Given $N$, $h$, $\Delta$, $\beta$, $m_1$ and $m_2$, solve:
\begin{IEEEeqnarray}{rCl}
{k_1}^{*}, {k_2}^{*} = \arg \min_{\{k_{1}, k_2\}} \underset{g\in \mathcal{G}_{\beta}}{\mathbb{E}} [\Pr(E_{fp}^{(1,2)})],
\end{IEEEeqnarray}
s.t. $1\leq k_1\leq m_1$, $1\leq k_2\leq m_2$, and $\mathcal{G}_{\beta}$ represents the set of all valid mappings $g(\cdot)$ between the IDs of the forwarding nodes and $\Delta$ satisfying the constraints in Section \ref{sec:Routing Constraints}.
\end{problem_stmt}}
 
The recovery process, detailed in Section \ref{sec:recovery}, involves path restoration from $\textbf{BF}_1$ followed by segment ID recovery from $\textbf{BF}_2$. Let $E_{fp}^{(1)}$ denote the false positive event from $\textbf{BF}_1$ wherein multiple paths of hop-length $h$ are recovered. As a consequence, the RSU declares an event $E_{fp}^{(1, 2)}$ when $E_{fp}^{(1)}$ occurs; otherwise, it proceeds to verify the locations of all the forwarding nodes recovered from $\textbf{BF}_{1}$. When using the path recovered from $\textbf{BF}_{1}$, let $E_{fp}^{(2)}$ denote the event when more than one segment ID sequence satisfying the constraints in Section \ref{sec:Routing Constraints} are recovered from $\textbf{BF}_{2}$. Thus, using $E_{fp}^{(1)}$ and $E_{fp}^{(2)}$, the overall false positive probability $\Pr(E_{fp}^{(1,2)})$ is given by\looseness=-1
\begin{IEEEeqnarray*}{rcl}
\Pr(E_{fp}^{(1,2)}) = \Pr(E_{fp}^{(1)}) + (1-\Pr(E_{fp}^{(1)})) \times \Pr(E_{fp}^{(2)}),
\end{IEEEeqnarray*}
where $\Pr(E_{fp}^{(1)})$ denotes the false positive probability of $\textbf{BF}_1$, and $\Pr(E_{fp}^{(2)})$ denotes the false positive probability of $\textbf{BF}_2$ conditioned on no false positive event from $\textbf{BF}_{1}$. Note that if the size of $\textbf{BF}_1$ is large enough, and the number of Hash functions in it is already optimized,\footnote{Bloom filters for embedding the edges of a path are well studied for large networks in \cite{amogh, suraj}. Therefore, subsequent solutions for optimizing the number the Hash functions $k_{1}$ can be found therein. By virtue of using existing edge embedding methods, our protocol is also applicable for large networks as $\textbf{BF}_{1}$ absorbs the scalability feature.} then $\Pr(E_{fp}^{(1)}) \approx 0$. As a result, we can approximate $\Pr(E_{fp}^{(1,2)}) \approx \Pr(E_{fp}^{(2)})$ under such assumptions.\looseness=-1

Using the above-mentioned relaxation, we propose a variant of Problem \ref{prob1}, which only focuses on optimizing the parameters of $\textbf{BF}_{2}$.\looseness=-1
{\begin{problem_stmt}{}\label{problem2}
Given $N$, $h$, $\Delta$, $m_1$, $m_2$ and $\beta$ solve:
\begin{IEEEeqnarray}{rCl}
{k_2}^{*} = \arg \min_{\{k_2\}} \underset{g\in \mathcal{G}_{\beta}}{\mathbb{E}} [\Pr(E_{fp}^{(2)})],
\end{IEEEeqnarray}
 s.t. $1\leq k_2\leq m_2$, and $\mathcal{G}_{\beta}$ represents the set of all valid mappings $g(\cdot)$ between the node IDs and $\Delta$ satisfying the constraints in Section \ref{sec:Routing Constraints}.\looseness=-1
\end{problem_stmt}}

Problem \ref{prob1} is challenging in general as it involves two Bloom filters. Given that the Bloom filter for embedding edge is already well understood in \cite{amogh},\cite{suraj}, we take the approach of solving Problem \ref{problem2}. Towards solving the above optimization problem, in the next section, we present a method to compute $\underset{g\in \mathcal{G}_{\beta}}{\mathbb{E}} [\Pr(E_{fp}^{(2)})]$, when the spatial distribution of the forwarding nodes is uniformly distributed in the set $\mathcal{G}_{\beta}$.\looseness=-1

\section{Optimization of $\textbf{BF}_2$ Parameters} \label{sec:optimization_of_BF2}
 In this section, we will first derive the expression for the objective function $\underset{g\in \mathcal{G}_{\beta}}{\mathbb{E}} [\Pr(E_{fp}^{(2)})]$, followed by presenting some low-complexity algorithms to calculate the same. After that, we will present a near-optimal solution for Problem \ref{problem2}.\looseness=-1
\subsection{Computing the Expression for $\underset{g\in \mathcal{G}_{\beta}}{\mathbb{E}} [\Pr(E_{fp}^{(2)})]$ }
To solve Problem \ref{problem2}, first we need to calculate the objective function $\underset{g\in \mathcal{G}_{\beta}}{\mathbb{E}} [\Pr(E_{fp}^{(2)})]$. 
% To compute the objective function, we formally define the false positive event of $\textbf{BF}_{2}$.
% \begin{definition}
% Let $g \in \mathcal{G}$ and the hop-length $h$ are fixed. A false positive event $E_{fp}^{(2)}$, occurs when a single path is recovered from $\textbf{BF}_1$. However, more than one valid segment sequence satisfying the routing constraints in Section \ref{sec:Routing Constraints} are recovered from $\textbf{BF}_2$.
% \end{definition}
Let us assume that for a given spatial-distribution $g \in \mathcal{G}_{\beta}$ of forwarding nodes, a single path of hop length $h$ is already recovered from $\textbf{BF}_1$, and the nodes of this path have been paired with various segment IDs to recover the locations of the nodes. Thus, from the first principle,
\begin{IEEEeqnarray}{rCl}
\label{eqn:obejctive1}
    \underset{g\in \mathcal{G}_{\beta}}{\mathbb{E}} [\Pr(E_{fp}^{(2)})] = \sum\limits_{g \in \mathcal{G}_{\beta}} \Pr(g)\Pr(E_{fp}^{(2)}|g),
\end{IEEEeqnarray}
where $\Pr(E_{fp}^{(2)}|g)$ denotes the false positive probability conditioned that the spatial distribution of the forwarding nodes is $g$ and $\Pr(g)$ denotes the probability of $g$. Assuming the spatial-distribution of the forwarding nodes $g$ to be uniform, i.e., $\Pr(g) = \frac{1}{|\mathcal{G}_{\beta}|}$, we can express our objective function as:
\begin{IEEEeqnarray}{rCl}
\label{eqn:objective2}
    \underset{g\in \mathcal{G}_{\beta}}{\mathbb{E}} [\Pr(E_{fp}^{(2)})] = \frac{1}{|\mathcal{G}_{\beta}|}\sum\limits_{g \in \mathcal{G}_{\beta}} \Pr(E_{fp}^{(2)}|g).
\end{IEEEeqnarray}

Given that $k_{2}$ Hash functions are used by every node, the number of bits lit in $\textbf{BF}_{2}$ through the journey of the packet is a random variable with a minimum value of $1$ and a maximum value of $min(m_2,k_2h)$. Thus, using $\alpha$ to denote the number of bits that have been lit in $\textbf{BF}_2$ by all the forwarding nodes, the false positive probability on the RHS of \eqref{eqn:objective2} is given in the following proposition.\looseness=-1

\begin{proposition} \label{prop:Pr_Efp} The probability of a false positive event for a given spatial-distribution of forwarding nodes, denoted by $\Pr(E_{fp}^{(2)}|g)$, can be written using the Bayes' rule as:
\begin{IEEEeqnarray}{rCl}
\label{Eq:Pr_Efp2}
    \Pr(E_{fp}^{(2)}|g) = \sum\limits_{\alpha=1}^{\min(m_{2},k_{2}h)}\Pr(E_{fp}^{(2)}|\alpha,g)\Pr(\alpha),
\end{IEEEeqnarray}
where $\Pr(E_{fp}^{(2)}|\alpha,g)$ denotes the false positive probability conditioned that $\alpha$ bits are lit in $\textbf{BF}_{2}$, and $\Pr(\alpha)$ denotes the probability of $\alpha$ bits lit in $\textbf{BF}_2$, which is given by \cite{suraj}:
\begin{IEEEeqnarray}{rCl}
\label{eqn:Pr_alpha}
\Pr(\alpha)=\frac{{m_2 \choose \alpha}  \sum\limits_{\gamma=0}^{\alpha}(-1)^{\gamma} {{\alpha}\choose{\gamma}}(\alpha-\gamma)^{k_2 h}}{m_2^{k_2 h}}.
    \end{IEEEeqnarray}
\end{proposition}
Using \eqref{Eq:Pr_Efp2} in \eqref{eqn:objective2}, the objective function of Problem \ref{problem2} can be represented as:
\begin{small}
\begin{IEEEeqnarray}{rCl}
\label{eqn:objective3}
    \underset{g\in \mathcal{G}_{\beta}}{\mathbb{E}} [\Pr(E_{fp}^{(2)})] = \frac{1}{|\mathcal{G}_{\beta}|}\sum\limits_{g \in \mathcal{G}_{\beta}} \sum\limits_{\alpha=1}^{\min(m_{2},k_{2}h)} \Pr(E_{fp}^{(2)}|\alpha,g) \Pr(\alpha),
\end{IEEEeqnarray}
\begin{IEEEeqnarray}{rCl}
\label{eqn:objective4}
     = \frac{1}{|\mathcal{G}_{\beta}|} \sum\limits_{\alpha=1}^{\min(m_{2},k_{2}h)} \Pr(\alpha)  \sum\limits_{g \in \mathcal{G}_{\beta}}  \Pr(E_{fp}^{(2)}|\alpha,g) ,
\end{IEEEeqnarray}
\begin{IEEEeqnarray}{rCl}
\label{eqn:objective5}
     =  \sum\limits_{\alpha=1}^{\min(m_{2},k_{2}h)}  \Pr(\alpha) \bigg[\frac{1}{|\mathcal{G}_{\beta}|} \sum\limits_{g \in \mathcal{G}_{\beta}}  \Pr(E_{fp}^{(2)}|\alpha,g)\bigg].
\end{IEEEeqnarray}
\end{small}
In the rest of the section, we provide an expression for $\frac{1}{|\mathcal{G}_{\beta}|} \sum\limits_{g \in \mathcal{G}_{\beta}}  \Pr(E_{fp}^{(2)}|\alpha,g)$. This expression has two parts, i.e., the total number of $h$-hop sequences of the spatial-distribution of forwarding nodes, denoted as $|\mathcal{G}_{\beta}|$, and the summation of probability of false-positive events of $\textbf{BF}_2$ for a given distribution $g$ when $\alpha$ bits are lit in $\textbf{BF}_2$. Next, we discuss the preliminaries to calculate the value of $|\mathcal{G}_{\beta}|$.\looseness=-1  
% Before we present our algorithm to compute $\Pr(E_{fp}^{(2)}|\alpha,g)$, we discuss preliminaries related to the false positive events associated with Bloom filters. 
% For the validity of the statement, please refer to.

\begin{definition}\label{def:valid_sequecnes}
 To enumerate the false positive events in $\textbf{BF}_2$, we define a valid segment sequence as the set of ordered locations that satisfy the communication constraints, as defined in Section \ref{sec:Routing Constraints}. In simple terms, a valid sequence of segments is a sequence from $\Delta$ such that the indexes of segments are in non-decreasing order with the condition that the difference between consecutive indexes is not more than $\beta$.\looseness=-1 
\end{definition}
For instance, for a hop-length $h = 4$, $\beta =3$ and $|\Delta| = 7$, a valid sequence is $\{A_{1}, A_{2} ,A_{3}, A_{5}, A_{7}\}$, and not $\{A_{1}, A_{2}, A_{3}, A_{3}, A_{7}\}$.
In general, we denote the set of all valid sequences of segments as $\mathcal{P}_{\beta}$. Note that, there is a one-to-one correspondence between the set $\mathcal{G}_{\beta}$ and the set $\mathcal{P}_{\beta}$, where one is the mapping of the node IDs to the segment IDs while the other is the ordered sequence of segment IDs. Therefore, for the communication constraint $\beta$, the value of $|\mathcal{G}_{\beta}| = |\mathcal{P}_{\beta}|$. The following lemma presents the result on the cardinality of $\mathcal{P}_{\beta}$ as a function of $|\Delta|,~\beta$ and $h$.\looseness=-1  

\begin{lemma} \label{lemma:path}
If $\mathcal{P_{\beta}}$ denotes the set of all valid sequences of segments of hop-length $h$ with the communication constraint $\beta$, then we can enumerate the elements of $\mathcal{P_{\beta}}$ for a given hop-length $h$ and the number of division of area $|\Delta|=r$ by constructing a $\beta$-ary tree.\looseness=-1
\end{lemma}
\begin{proof}
We utilize a $\beta$-ary tree to enumerate the correct segment sequences while adhering to the communication constraints outlined in Section \ref{sec:Routing Constraints}, and considering the condition that the maximum value of $|\Delta|$ is $r$.
A $\beta$-ary tree is a hierarchical data structure used for generating non-decreasing sequences. Let the value of the \emph{main root node} of $\beta$-ary tree be one. Every root node with value $X \in [r]$ has $\beta+1$ leaf nodes, where the value of leaf nodes varies from $X+0$ to $X+\beta$. This structure ensures a systematic and incremental progression of values, facilitating organized data representation. Specifically, for a given root node $X$, the associated leaf nodes will have values $\{X+0, X+1,\ldots,X+\beta\}$. For instance, with $\beta=2$, the main root node $X=1$ will have leaf nodes with values $\{1,2,3\}$, and a subsequent root node with $X=2$ will have leaf nodes with values $\{2,3,4\}$. Since we are forming valid sequences for an $h$-hop path, the maximum depth of this $\beta$-ary tree will be $h$. From this tree, we select only those stems whose leaf nodes have values at most $r$. These stems are then added to the set $\mathcal{P_{\beta}}$. This completes the proof.\looseness=-1
% , where $\mathcal{PATHS}$ is a set that stores all the valid sequences of hop length $h$. This problem is similar to the climbing staircase problem, where an individual aims to reach the $A_{i}^{th}$ stair with $h$ steps and can climb a minimum of zero steps and a maximum of $\beta$ steps at a time, with $i \in [1,r]$.\looseness=-1
\end{proof}

\begin{corollary} \label{lemma:pathbeta}
With $\mathcal{P}_{\beta}$ obtained from Lemma \ref{lemma:path}, we get $|\mathcal{P_{\beta}}|$ trivially.\looseness=-1 
% given by Algorithm \ref{al:path_calculation}
\end{corollary}
\begin{figure}
    \centering
    \includegraphics[scale=0.4]{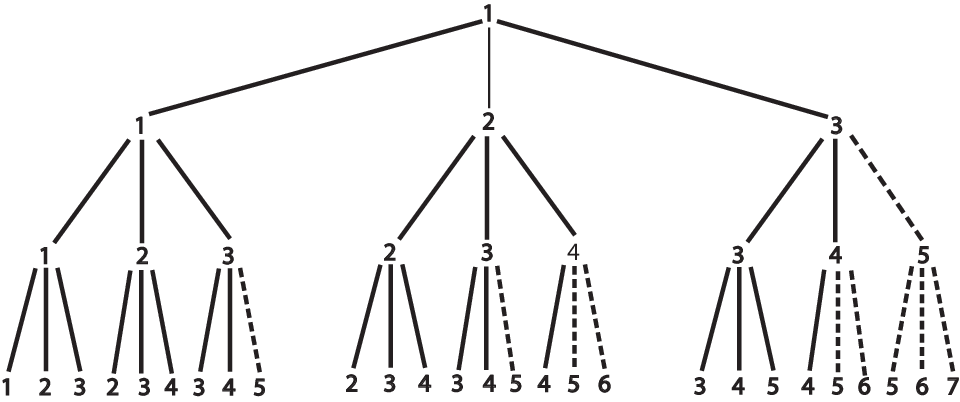}
    \caption{ A representation of the $\beta$-ary tree for the network configuration of $\beta=2,~|\Delta|=r=4$ and $h=3$. The solid line shows the valid segment sequences, and the dotted lines show the invalid segment sequences. Stems that are in solid lines from the root to the leaf nodes are added to the set $\mathcal{P_{\beta}}$.\looseness=-1
    % Depiction of $\beta$-ary tree for the network configuration of $\beta=2,~|\Delta|=r=4,~h=3$. The solid line shows the valid segment sequences, and the dotted lines show the invalid segment sequences. Stems that are in solid lines from the root to the leaf nodes are added into the set $\mathcal{P_{\beta}}$.
    }
    \label{fig:beta-ary tree}
    % \vsapce{-0.3cm}
    \vspace{-0.4cm}
\end{figure}

\begin{table*}[h!] 
\caption{ Illustrating the computation of false positive events as a function of extra recovery.}
\centering
 \begin{tabular}{|c | c |  c | c| c|} 
 \hline
 Valid Paths  & Extra 1 Recovery  & Extra 2 Recovery  & \ldots & Extra $(|\Delta|)-1)h$ Recovery\\
 ($\mathcal{P}_{\beta}$) & ($\mathcal{R} =1$) & ($\mathcal{R} =2$) & & ($\mathcal{R} =(|\Delta|-1)h$)
 \\[0.5ex] 
 \hline\hline
 1111\ldots$1_{h}$ & $C_{1,1}$ & $C_{1,2}$ & \ldots & $C_{1,(|\Delta| -1)h}$ \\ 
 1111\ldots$2_{h}$ & $C_{2,1}$ & $C_{2,2}$ & \ldots &  $C_{2,(|\Delta| -1)h}$ \\
 . & . & . & \ldots & . \\
 . & . & . & \ldots & . \\
. & . & . & \ldots  & . \\
 123\ldots $\Delta$ \ldots$\Delta_{h}$ & $C_{|\mathcal{P}_{\beta}|,1}$ & $C_{|\mathcal{P}_{\beta}|,2}$ & \ldots &  $C_{|\mathcal{P}_{\beta}|,(|\Delta| -1)h}$\\ [1ex] 
 \hline\hline
 Total & $C_1 = \sum_{x=1}^{|\mathcal{P}_{\beta}|}C_{x,1} $ &  $C_2 = \sum_{x=1}^{|\mathcal{P}_{\beta}|}C_{x,2} $ & \ldots &  $C_{(|\Delta| -1).h}= \sum_{x=1}^{|\mathcal{P}_{\beta}|}C_{x,(|\Delta| -1)h}$\\
 
 % Total & $total_1\times p_1^1 \times p_2^{(|\Delta| -1).h-1} $ &  $total_2\times p_1^2\times p_2^{(|\Delta| -1).h-2} $ & \ldots &  $total_{(|\Delta| -1).h}\times p_1^{(|\Delta| -1).h}\times p_2^{0} $\\
 
 \hline
 \end{tabular}
 \label{T:falsepositiveCalculation}
\end{table*}

 Now, we proceed towards the calculation of the expression $\Pr(E_{fp}^{(2)}|\alpha,g)$ for which we formally define the collision event in Bloom filter.\looseness=-1 

\begin{definition}\label{def:collision} A collision event in a Bloom filter is defined as an event when we verify an object from the Bloom filter that was not initially embedded in it.\looseness=-1
\end{definition}

Given a Bloom filter of size $m$ bits with $\alpha$ bits already lit in it, the probability of a collision event is given by $(\alpha/m)^{k}$, where $k$ is the number of Hash functions used to embed every object. With respect to $\textbf{BF}_{2}$, when we verify the segment ID of a node, we denote the probability of a collision event as $p_1$, and the probability of non-collision events as $p_2$, which are given as\looseness=-1
\begin{IEEEeqnarray*}{rcl}
   p_1 = (\alpha/m_2)^{k_2}, \quad  p_2 = 1 - p_1.
\end{IEEEeqnarray*}
Note that the above-discussed collision events are crucial in enumerating the false positive events in $\textbf{BF}_{2}$. For instance, with hop-length $h = 4$, $\beta =3$ and $|\Delta| = 7$, a valid sequence $\{A_{1}, A_{2}, A_{3}, A_{5}, A_{7}\}$ will result in false positive event if there is collision of the index values of the event $\{A_6\}$ associated with the $1^{st}$ node or the $2^{nd}$ node in the multipath.\looseness=-1

Based on Section \ref{sec:Bloom Filters}, during recovery of the node-segment pairs from the Bloom filter $\textbf{BF}_2$, we check all possible combinations of node-segment pairs in the Bloom filter. Therefore, the RSU verifies $h|\Delta|$ pairs in $\textbf{BF}_2$, out of which only $h$ were embedded. We define the set of remaining pairs as false pairs, as formally defined below.\looseness=-1
\begin{definition} \label{def:setF}
For a given valid segment sequence in $\mathcal{P}_{\beta}$, the set of node-segment pairs that were not originally embedded in $\textbf{BF}_2$, however, is checked for membership in $\textbf{BF}_2$ at the RSU, is called the set of false pairs, denoted by $\mathcal{F}$, where $|\mathcal{F}| = h(|\Delta|-1)$.\looseness=-1 
\end{definition}
We want to emphasise that the set $\mathcal{F}$ will be different for each valid sequence in $\mathcal{P}_{\beta}$. However, $|\mathcal{F}|$ is the same for every set $\mathcal{F}$.\looseness=-1
\begin{definition} \label{def:setR}
For a given valid segment sequence in $\mathcal{P}_{\beta}$, the set of node-segment pairs recovered from $\textbf{BF}_2$, not on the original segment sequence, is called extra recoveries, denoted by $\mathcal{R}$, where $\mathcal{R} \subseteq \mathcal{F}$.\looseness=-1
\end{definition}
Note that, although $|\mathcal{F}|$ is same for every sequence in $\mathcal{P}_{\beta}$, set $\mathcal{R}$ is a random subset of $\mathcal{F}$, which depends on the collision event defined in Definition \ref{def:collision}.\looseness=-1
\begin{example}
% Let us endeavour to comprehend the occurrence of a false positive event through an example. Suppose we receive only one extra recovery from $\textbf{BF}_2$, meaning $|\mathcal{R}|=1$. It is important to note that there are $|\mathcal{F}|$ possible ways to obtain one extra recovery, but not all of these $|\mathcal{F}|$ pairs will lead to a false positive event. For instance, 
Consider a segment sequence for $\beta =2$, $h=3$ and $r=7$ as follows: $[(I_{1},A_{1})\rightarrow (I_{2},A_{2})\rightarrow (I_{3},A_4)\rightarrow (I_{4},A_6) ]$. For this sequence, $\mathcal{F} = \{I_1,I_2,I_3,I_4\}\times \Delta \backslash \{(I_{1},A_{1}),(I_{2},A_{2}),(I_{3},A_4), (I_{4},A_6)\}$. If $\mathcal{R}=\{(I_4, A_7)\}$ or $\mathcal{R}=\{(I_4, A_2)\}$, this will not result in a false positive event since substituting the extra recovery in the original path would violate the node communication constraint in Section \ref{sec:Routing Constraints}. However, if $\mathcal{R} = \{(I_4,A_5)\}$ or $\mathcal{R} = \{(I_4,A_4)\}$, then replacing this with $(I_4, A_6)$ will lead to a false positive event.
\end{example}

% To count the false positive events, suppose that the nodes are distributed in the network with the location sequence lying in the set $\mathcal{P}_{\beta}$. 
For a given segment sequence in $\mathcal{P}_{\beta}$, a false positive event occurs when either one of the other sequences in $\mathcal{P}_{\beta}$ is recovered from the Bloom filter. Since other sequences in $\mathcal{P}_{\beta}$ are recovered through extra recoveries, we need to relate extra recoveries with the false positive events. During the recovery process from $\textbf{BF}_{2}$, the number of extra recoveries is a random variable. Therefore, we need to consider all possible cases of extra recoveries such that $|\mathcal{R}| \leq |\mathcal{F}|$. Assuming $\mathcal{P}_{\beta} = \{\mathbf{p}_{1}, \mathbf{p}_{2}, \ldots, \mathbf{p}_{|\mathcal{P}_{\beta}|}\}$, suppose that $\mathbf{p}_{x}$ for some $x \in [1, |\mathcal{P}_{\beta}|]$, is the true sequence of segments. If the Bloom filter is lit based on this sequence $\mathbf{p}_{x}$ from the $h$ nodes, let $C_{xj}$ represent the number of sequences in $\mathcal{P}_{\beta}\backslash \{\mathbf{p}_{x}\}$ which can be recovered from the Bloom filter when $j = |\mathcal{R}|$ extra recoveries are lit in the Bloom filter. If $C_{xj}$ can be computed for all $j \in [1, |\mathcal{F}|]$, then the false positive probability for the sequence $\mathbf{p}_{x}$ can be computed. Let the sum of all the number of sequences across all the paths in $\mathcal{P}_{\beta}$ for a given number of extra recovery $|\mathcal{R}|=j$ is represented as $C_{j}$, where $C_{j} \triangleq \sum_{x = 1}^{|\mathcal{P}_{\beta}|}C_{xj}$. Then, we use the following theorem to define the probability of the event of a false positive on the parameter $C_{x,j}$.\looseness=-1
\begin{theorem}
Given the spatial-distribution of the forwarding nodes $g$, where the packet traversed through $h$ hops, the closed-form expression for the probability of the false positive event of $\textbf{BF}_2$ is:\looseness=-1

% \begin{small}\begin{IEEEeqnarray*}{rcl}
% \label{overall_fp2_theorem}
% \frac{1}{|\mathcal{G}_{\beta}|}\sum\limits_{g \in \mathcal{G}_{\beta}}\Pr(E_{fp}^{(2)}|\alpha,g) = \frac{1}{|\mathcal{G}_{\beta}|}\sum_{j=1}^{h(|\Delta|-1)} p_1^j\times p_2^{h(|\Delta|-1)-j}\times \Big(\sum_{x = 1}^{|\mathcal{G}_{\beta}|}C_{xj}\Big),
% \end{IEEEeqnarray*}\end{small}
% {\scriptsize
% \begin{IEEEeqnarray*}{rCl}
% \frac{1}{|\mathcal{G}_{\beta}|}\sum\limits_{g \in \mathcal{G}_{\beta}}\Pr(E_{fp}^{(2)}|\alpha,g) &=& \frac{1}{|\mathcal{G}_{\beta}|}\sum_{j=1}^{h(|\Delta|-1)} p_1^j p_2^{h(|\Delta|-1)-j} \left(\sum_{x = 1}^{|\mathcal{G}_{\beta}|}C_{xj}\right)
% \end{IEEEeqnarray*}
% }
\begin{footnotesize}
\begin{IEEEeqnarray*}{rcl}
\label{overall_fp2_theorem}
\frac{1}{|\mathcal{G}_{\beta}|}\sum\limits_{g \in \mathcal{G}_{\beta}}\Pr(E_{fp}^{(2)}|\alpha,g) = \frac{1}{|\mathcal{G}_{\beta}|}\sum_{j=1}^{h(|\Delta|-1)} p_1^j\times p_2^{h(|\Delta|-1)-j}\times \Big(\sum_{x = 1}^{|\mathcal{G}_{\beta}|}C_{xj}\Big),
\end{IEEEeqnarray*}
\end{footnotesize}

\noindent\noindent where $p_1 = (a/m_2)^{k_2}$ and $p_2=1-p_1$, and $|\mathcal{G}_{\beta}|$ is taken from Lemma \ref{lemma:path}.
\end{theorem}
\begin{proof}
When a packet traverses through a $h$-hop path, then there will be embedding of $h$ segment IDs in $\textbf{BF}_2$. However, due to the properties of the Bloom filter, assume that there are $j$ extra recoveries from $\textbf{BF}_2$ as defined in Definition \ref{def:setR}. The probability of the recovering $j$ elements out of the set $\mathcal{F}$ and not recovering the $(|\mathcal{F}|-j)$ elements is $p_{1}^{j}.p_{2}^{|\mathcal{F}|-j}$.
Due to $j$ extra recoveries, a given valid segment sequence $\mathbf{p}_x, ~x \in [1, |\mathcal{P}_{\beta}|]$, may result into one or more path such as $\mathbf{p}_{x_1}, \mathbf{p}_{x_2}, \ldots,\mathbf{p}_{x_{C_{x_j}}}$, i.e., $E_{fp}^{(2)}|\alpha,\mathbf{p}_x=\mathbf{p}_{x_1} \bigcup \mathbf{p}_{x_2}\bigcup \ldots\bigcup\mathbf{p}_{x_{C_{x_j}}}$. However, the occurrence of all the extra paths is an independent and mutually exclusive event. Therefore, $\Pr(E_{fp}^{(2)}|\alpha,\mathbf{p}_x)=\Pr(\mathbf{p}_{x_1}) + \Pr(\mathbf{p}_{x_2})+\ldots+\Pr(\mathbf{p}_{x_{C_{x_j}}})$.
Using Definition \ref{def:collision}, Definition \ref{def:setF} and Definition \ref{def:setR}, the probability of false positive event can be represented as

\begin{small}\begin{IEEEeqnarray*}{rcl}
\Pr(E_{fp}^{(2)}|\alpha, g)= \Pr(E_{fp}^{(2)}|\alpha, \mathbf{p}_{x}) = \sum_{j=1}^{h(|\Delta|-1)} p_1^j\times p_2^{h(|\Delta|-1)-j}\times C_{xj}.
\end{IEEEeqnarray*}\end{small}

Thus, the false positive probability averaged over all possible true segment sequences can be computed as
\begin{IEEEeqnarray*}{rcl}
\frac{1}{|\mathcal{G}_{\beta}|}\sum\limits_{g \in \mathcal{G}_{\beta}}\Pr(E_{fp}^{(2)}|\alpha,g) = \frac{1}{|\mathcal{P}_{\beta}|} \sum_{x = 1}^{|\mathcal{P}_{\beta}|}\Pr(E_{fp}|\alpha, \mathbf{p}_{x}),
\end{IEEEeqnarray*}
\begin{IEEEeqnarray}{rcl}
\label{overall_fp}
 = \frac{1}{|\mathcal{P}_{\beta}|}\sum_{j=1}^{h(|\Delta|-1)} p_1^j\times p_2^{h(|\Delta|-1)-j}\times (\sum_{x = 1}^{|\mathcal{P}_{\beta}|}C_{xj}),
\end{IEEEeqnarray}
where uniform distribution is assumed on $\mathbf{p}_{x}$. We can express the above equation with $\mathcal{G}_{\beta}$ as 
\begin{IEEEeqnarray}{rcl}
\label{overall_fp2}
 = \frac{1}{|\mathcal{G}_{\beta}|}\sum_{j=1}^{h(|\Delta|-1)} p_1^j\times p_2^{h(|\Delta|-1)-j}\times (\sum_{x = 1}^{|\mathcal{G}_{\beta}|}C_{xj}).
\end{IEEEeqnarray}

Using the above expression, given the parameters of Bloom filter, the expression for $\Pr(E_{fp}^{(2)}|\alpha,g)$ can be computed provided $C_{j} \triangleq \sum_{x = 1}^{|\mathcal{P}_{\beta}|}C_{xj}$ values are taken from Table \ref{T:falsepositiveCalculation}. Next, we present an example of false positive event followed by a method to fill the values in Table \ref{T:falsepositiveCalculation}.\looseness=-1 
 % compute $\{C_{j}, j = 1, 2, \ldots, h(|\Delta|-1)\}$. Once we present this method, we can use its outcome and apply them in \eqref{overall_fp}. 
\end{proof}

One method to calculate the value of $C_{x,j}$ in Table \ref{T:falsepositiveCalculation} is an exhaustive search approach. For exhaustive search, suppose for a given $h$-hop  sequence $\mathbf{p}_x,x\in[1,|\mathcal{P}_{\beta}|]$, we have $|\mathcal{R}|=1$ extra recovery. 
To calculate the number $C_{x,1}$, we have to do a comparison  $\binom{h}{1} |\mathcal{F}|$ times with complexity order $\mathcal{O}(h(r-1))$ to verify the communication constraint. Similarly, for $j$ extra recoveries, we have to do a comparison $\binom{h}{j}|\mathcal{F}|^{j}$ times. Therefore, if we sum all the possible scenarios of extra recovery for a given sequence $\mathbf{p}_x$, we have $\binom{h}{1} |\mathcal{F}| + \binom{h}{2} |\mathcal{F}|^{2} + \binom{h}{3} |\mathcal{F}|^{3}+\ldots+\binom{h}{h} |\mathcal{F}|^{h} = (1+h(|\Delta|-1))^{h}-1\geq 2^h$. Discounting the complexity of the DFS algorithm to verify a valid sequence, the complexity of the discussed exhaustive search is exponential with $\mathcal{O}((1+h(|\Delta|-1))^{h})$.\looseness=-1 

To further reduce the complexity, we can also apply the communication constraint $\beta$ in the exhaustive search discussed above. Therefore, any position in $h$-hop sequence can be replaced by at most  $\beta$ values. Therefore the complexity of calculating $C_{x,1}$ is $\mathcal{O}(\beta h) < \mathcal{O}(h (r-1))$. Consequently, the complexity of the exhaustive search to find all the $C_{x,j}$ values for a row in Table \ref{T:falsepositiveCalculation} with $\beta$ constraint is $\mathcal{O}((1+\beta))^{h})$, which is still exponential in $h$. Due to the above discussion, in the next section, we will present a method to populate the Table \ref{T:falsepositiveCalculation} with low complexity.\looseness=-1

\subsection{A Low-Complexity Method to Compute $\underset{g\in \mathcal{G}_{\beta}}{\mathbb{E}} [\Pr(E_{fp}^{(2)})]$}
\label{sec:Low_complexity_C_j_any_beta}
To calculate the value of $C_{x,1}$, we propose a low-complexity algorithm with the help of the following proposition. 
\begin{proposition}\label{prop:value_of_C_x_1}
Given there is only $|\mathcal{R}|=1$ extra recovery for a given valid segment sequence $\mathbf{p}_x,~x\in[1,|\mathcal{P}_{\beta}|]$, we can obtain the value of $C_{x,1}$ using Algorithm \ref{al:algo2} for any $\beta$.
\end{proposition}
\begin{proof}
We propose Algorithm \ref{al:algo2} to calculate $C_{x,1}$. Algorithm \ref{al:algo2}, simply takes the difference of extreme points in a window of size $t=3$, and keeps on adding this difference for the path of hop-length $h$. The time-complexity of Algorithm \ref{al:algo2} is $\mathcal{O}(h)$, as there will be $h$ additions.\looseness=-1

As explained above, we will calculate the false positive paths possible for each path $\mathbf{p}_x,~x\in [1,|\mathcal{P}_{\beta}|]$ using Algorithm \ref{al:algo2}. After this, we will obtain $C_1$ by summing it over all the paths, which is given by $C_1 = \sum_{x=1}^{|\mathcal{P}_{\beta}|}C_{x,1}$.\looseness=-1
% The complexity for calculating $C_1$ is $\mathcal{O}(h|\mathcal{P}_{\beta}|)$. \bl{The value of $C_1$ is calculated as $C_1 = \sum_{x=1}^{|\mathcal{P}_{\beta}|}C_{x,1}$ for $x \in |\mathcal{P}_{\beta}|$, where all valid paths in set $\mathcal{P}_{\beta}$ are obtained through a $\beta$-ary tree. The time complexity of obtaining all the valid paths in $\mathcal{P}_{\beta}$ through a $\beta$-ary tree is $\mathcal{O}((\beta+1)^h)$, which is exponential in $h$. Therefore, while the time complexity over an already enumerated set of paths is polynomial, i.e., $\mathcal{O}(h|\mathcal{P}_{\beta}|)$, the time complexity of generating all valid paths is exponential in $h$, which is a one-time process for a given $h$, $\beta$, and $r$. However, this represents a worst-case scenario in networks where the hop length $h$ is very large.}
\end{proof}
Further, we will explain the method of computing $C_j$ for $j \geq 2$.\looseness=-1
\begin{algorithm}
\caption{Algorithm to count the values of $C_{x,1}$}
    \label{al:algo2}
\begin{algorithmic}[1]
\Require 
\Statex \textbf{1)} Communication constraint $\beta$
\Statex  \textbf{2)} Hop length of given path $h$
\Statex \textbf{3)} A set of segment IDs $\Delta$
\Statex \textbf{4)} A path $\mathbf{p}_x$, from the set of all non-decreasing paths $\mathcal{P}_{\beta}$
\Ensure $C_{x,1}$

\State Initialize $C_{x,1} \Leftarrow 0$
\State Initialize the variable $d$, which stores the difference between two numbers, $d\Leftarrow 0$

\For{$i$ from $0$ to $h-1$}
    \State $d \leftarrow \mathbf{p}_x[i+2]-\mathbf{p}_x[i]$\;
    
     \If {$0\leq d < \beta$}
    \State $C_{x,1} \leftarrow C_{x,1} + d$\;
    
    \ElsIf {$\beta \leq d \leq 2\beta$}
    \State $C_{x,1} \leftarrow C_{x,1} + 2\beta -d$\;
    
    \EndIf
    \State $C_{x,1} \leftarrow C_{x,1} -\min(|\Delta| - \mathbf{p}_x[h],\beta)$\;
\EndFor    
% \If{$N$ is even}
%   \State $X \Leftarrow X \times X$
%   \State $N \Leftarrow \frac{N}{2} $  \Comment{This is a comment}
% \ElsIf{$N$ is odd}
%   \State $y \Leftarrow y \times X$
%   \State $N \Leftarrow N - 1$
% \EndIf
% \EndWhile
\end{algorithmic}
\end{algorithm}

% To calculate the the number $C_{x,j}$, where $j\geq 2$ and $x \in [1,|\mathcal{P}_{\beta}|]$, we can also go with the exhaustive search. For $|\mathcal{R}|=j=2$, i.e., for extra two recoveries, we have to use the window of size $t=4$. This window of size $t=4$ will be broken further in small sub-windows of size $t=3$. There will be three sub-windows of size $t=3$. The same process as explained above for a window of size $t=3$ will be followed to calculate false positives. The time complexity of finding the false positives through the windows technique will be $\mathcal{O}(\beta^{t} h)$. 

% Besides having exponential complexity, the windows technique will not give us the exact count for the $j \geq 2$, as it gave for $j=1$. In the window technique for $j \geq 2$, we will not be able to count the cases when the extra recoveries are not consecutive. For illustration, take a correct segment sequence of six nodes and suppose node 2 and node 5 have an extra recovery.  However, if we use the window of size $t=4$, we will not be able to calculate the false positives because the window which starts from node 1, will be able to calculate the false positive possibilities for node 2. Similarly, if we take the window that starts from node 3, we can calculate the false positive possibilities for node 5. However, we will not be able to calculate those possibilities where node 2 and node 5 extra recoveries co-occur. Further, we will discuss a set-based calculation method with low time complexity.   

\begin{proposition}\label{prop:value extension any beta}
Given there are $|\mathcal{R}|=j, j\geq 2$ extra recoveries for a given valid segment sequence $\mathbf{p}_x,~x\in[1,|\mathcal{P}_{\beta}|]$, and we have the value of $C_{x,1}$ obtained from Algorithm \ref{al:algo2}, the value of  $C_{x,j}$, for $j\geq 2$, is given by: 
% Given $C_j$ for a segment sequence, the number of sets of false pairs such that $|\mathcal{R}| = j$, which will cause a false positive event, is given by:
\begin{IEEEeqnarray}{rCl}
\label{eq:cxj_beta_expression}
C_{x,j} \geq \sum_{l=1}^{C_{x,1}}\binom{h(|\Delta|-1) - l}{j-1}.
\end{IEEEeqnarray}
\end{proposition}

\begin{proof}
We derive this expression by sequentially selecting one of the $j$ pairs, and then selecting the remaining $j-1$ pairs from the remaining false pairs of $\mathcal{F}$, to construct the set $\mathcal{R}$. Subsequently, two pairs from the $j$ pairs are fixed, and then the rest of the $j-2$, pairs can be picked from the remaining false pairs of $\mathcal{F}$, to construct the set $\mathcal{R}$. This process is repeated up to the case of fixing $j$ pairs to fill $\mathcal{R}$, and then picking the rest of the false pairs from $\mathcal{F}$.\looseness=-1 
% The time complexity of calculating a binomial coefficient $\binom{n}{k}$ is $\mathcal{O}(k)$. Therefore the time complexity of calculating $\binom{h(|\Delta|-1) - l}{j-1}$ is $\mathcal{O}(j)$. The time complexity of \eqref{eq:cxj_beta_expression} will be $\mathcal{O}(hj)$.\looseness=-1
\end{proof}

Using the above result, we are now ready to have an expression for $\{C_{j}, j > 1\}$, given by
\begin{IEEEeqnarray}{rCl}
\label{eq:cj_beta_expression1}
C_j \geq \sum_{x=1}^{|\mathcal{P}_{\beta}|} \sum_{l=1}^{C_{x,1}}\binom{h(|\Delta|-1) - l}{j-1}.
\end{IEEEeqnarray}
From the above expression, it is clear that as long as we have $C_{x,j}$, we can compute $\{C_{j}\}$ for \bl{a} given $\beta$. The time complexity of calculating \eqref{eq:cj_beta_expression1} is $\mathcal{O}(j|\mathcal{P}_{\beta}|C_{1})$. By plugging the values of $C_{x,1}$ from Proposition \ref{prop:value_of_C_x_1} in \eqref{eq:cxj_beta_expression}, we obtain the lower bound on $C_{x,j}$, which is used in \eqref{overall_fp2} to obtain the the lower bound on the expression for $\underset{g\in \mathcal{G}_{\beta}}{\mathbb{E}} [\Pr(E_{fp}^{(2)})]$. Obtained lower bound of $\underset{g\in \mathcal{G}_{\beta}}{\mathbb{E}} [\Pr(E_{fp}^{(2)})]$ then plugged in \eqref{eqn:objective5} to obtain the lower bound on the overall average probability of false-positive event.\looseness=-1

\begin{remark}
We want to emphasize that \eqref{eq:cxj_beta_expression} counts a lot of scenarios \bl{that result} in false positives. However, it misses a few scenarios. Therefore, the value of $C_{x,j}$, for $j\geq 2$, given by \eqref{eq:cxj_beta_expression}, is less than the actual number of false positive events. For example, consider a valid segment sequence for hop-length $h=2$, $\beta=2$ and $|\Delta|=7$ as $\{A_1,A_3,A_5\}$. For an extra one recovery, the possibilities that can replace $A_3$ are none (refer to Algorithm \ref{al:algo2}). Similarly, using Algorithm \ref{al:algo2}, the values which can replace $A_5$ is two i.e., $\{A_3\}$ and $\{A_4\}$. Hence the value of $C_{x,1} = 2$. Further, to calculate the number of false positives for the extra two recoveries, i.e., $C_{x,2}$, if we use \eqref{eq:cxj_beta_expression} then we will get $\binom{2.(7-1)-1}{1} + \binom{2.(7-1)-2}{1} = 21$. The value of $C_{x,2}=21$ considers the cases where the last location, i.e., $A_5$, is replaced by $A_3$ or $A_4$. However, this expression fails to calculate the possibilities when the second position, i.e., $A_3$ is replaced by $A_2$, and then at the last location, we can have the possibility of replacement by $A_2$, $A_3$, or $A_4$. These replacements are only possible in cases where the extra recovery is greater than two.
\end{remark}

% We want to emphasise that despite undercounting in \eqref{eq:cxj_beta_expression}, the lowest value of the average false positive probability that we obtain by \eqref{eqn:objective5} coincides with the actual values, which will be depicted by results discussed in Section \ref{sec:results}.
 In the above section, we have presented an algorithm for calculating $C_j$ for any $\beta$. Next, we will obtain a solution for Problem \ref{problem2}.\looseness=-1
\subsection {On Solving Problem \ref{problem2}}
By plugging the values of $C_{x,1}$ from Proposition \ref{prop:value_of_C_x_1} and by plugging the lower bound values of $C_{x,j}$ from \eqref{eq:cxj_beta_expression} into \eqref{overall_fp2}, we obtain the lower bound for $\underset{g\in \mathcal{G}_{\beta}}{\mathbb{E}} [\Pr(E_{fp}^{(2)}|\alpha,g)]$. Then, the obtained lower bound value of $\underset{g\in \mathcal{G}_{\beta}}{\mathbb{E}} [\Pr(E_{fp}^{(2)}|\alpha,g)]$ is plugged in \eqref{eqn:objective5} to obtain the lower bound on the overall average probability of false-positive events\footnote{This is a lower bound, and we do not provide a theoretical analysis on its tightness due to the analytical complexity. However, as demonstrated through simulation results in Section \ref{sec:results}, we show that the proposed lower bound performs well.}.
% As an alternate approach for $\beta=1$, we use \eqref{eqn:EDelta} and \eqref{eqn:EDeltaDash} to obtain the value in \eqref{eqn:f_j expanded}. Then, we use \eqref{eqn:f_j expanded} in \eqref{eq:cj_expression}, and in turn use them in \eqref{overall_fp2} to get the expressions for $\underset{g\in \mathcal{G}_{\beta}}{\mathbb{E}} [\Pr(E_{fp}^{(2)}|\alpha,g)]$ for $\beta=1$. Finally, we use the value of $\underset{g\in \mathcal{G}_{\beta}}{\mathbb{E}} [\Pr(E_{fp}^{(2)}|\alpha,g)]$ in \eqref{eqn:objective5} to obtain the average probability of false positive events.\looseness=-1 

To obtain a near-optimal solution to Problem \ref{problem2}, we use the computed lower bound on $\underset{g\in \mathcal{G}_{\beta}}{\mathbb{E}} [\Pr(E_{fp}^{(2)})]$ as the objective function in order to determine the values around $k_2^{*}$. Given that the optimization variable $k_2 \in [1,m_2]$, we need to evaluate the objective function for all $k_2 \in [1,m_2]$, and then find its minimum value. Towards that goal, as $k_2$ varies, we observe that we have to compute the terms $\Pr(\alpha)$  in \eqref{eqn:Pr_alpha} for every $\alpha \in [1,\min(m_2,k_2h)]$ and correspondingly the values of $p_1$ and $p_2$. Whereas the value of $C_{x,j}$ needs to be computed only once. 
% We want to emphasize that for the case of $\beta = 1$, we do not even need to use Algorithm \ref{al:algo2} to compute the values of $C_{x,1}, x \in \mathcal{P}_1$, as the values of $C_{x,1}$ can be computed easily with the expressions given in Section \ref{sec:c_j for beta 1}.
Once the value of $\underset{g\in \mathcal{G}_{\beta}}{\mathbb{E}} [\Pr(E_{fp}^{(2)})]$ is ready, we use gradient descent algorithm to find the optimal value of \eqref{eqn:objective5} by varying the value of $k_2$. In the next section, we will present the complexity analysis of our proposed algorithms to obtain the parameters of $\textbf{BF}_2$.\looseness=-1   

\subsection{Complexity Analysis}
Recall that in the previous section, we proposed a low-complexity algorithm to compute the optimal number of Hash functions that minimizes $\underset{g\in \mathcal{G}_{\beta}}{\mathbb{E}} [\Pr(E_{fp}^{(2)})]$ for given network parameters $h$, $\beta$, and $r$.
% Recall that in the previous section, for the given set of network parameters $h$, $\beta$, and $r$, we have provided a low-complexity algorithm to obtain the number of Hash functions that minimizes $\underset{g\in \mathcal{G}_{\beta}}{\mathbb{E}} [\Pr(E_{fp}^{(2)})]$. 
In this section, we quantify the computational overhead required to execute each step in the proposed method. Our method involves populating Table \ref{T:falsepositiveCalculation}, and subsequently using Table \ref{T:falsepositiveCalculation}, we derive the expressions for the lower bound on the overall average probability of false positives, and then solve for the minima over $k_2$. To populate the first column of table \ref{T:falsepositiveCalculation}, we enumerate the set of valid paths in $\mathcal{P}_{\beta}$ using a $\beta$-ary tree which has time complexity of $\mathcal{O}((\beta+1)^h)$, which is exponential in $h$ and it is a one-time process for given $h$, $\beta$, and $r$. After that, the second column of Table \ref{T:falsepositiveCalculation} is computed by evaluating $C_1 = \sum_{x \in \mathcal{P}_{\beta}} C_{x,1}$, where $C_{x,1}$ is calculated using Algorithm \ref{al:algo2} with time complexity $\mathcal{O}(h)$, resulting in total time complexity of $\mathcal{O}(|\mathcal{P}_{\beta}|h)$ which is polynomial in $h$ given $\mathcal{P}_{\beta}$ is already enumerated. Other columns of Table \ref{T:falsepositiveCalculation} are calculated using \eqref{eq:cj_beta_expression1} which requires the time complexity of $\mathcal{O}(jh|\mathcal{P}_{\beta}|C_{1}|\Delta|)$. Finally, for calculating the optimal value of $\underset{g\in \mathcal{G}_{\beta}}{\mathbb{E}} [\Pr(E_{fp}^{(2)})]$, we evaluate \eqref{eqn:objective5} over $[1,k_2]$ values using a gradient descent algorithm which has a time complexity of $\mathcal{O}(k_2)$. In the next section, we will present experimental results to validate the effectiveness of the above-defined solution.\looseness=-1 

% \bl{The time complexity of obtaining all the valid paths in $\mathcal{P}_{\beta}$ through a $\beta$-ary tree is $\mathcal{O}((\beta+1)^h)$, which is exponential in $h$. The value of $C_1$ is calculated as $C_1 = \sum_{x=1}^{|\mathcal{P}_{\beta}|}C_{x,1}$ for $x \in |\mathcal{P}_{\beta}|$, where all valid paths in set $\mathcal{P}_{\beta}$ are obtained through a $\beta$-ary tree. Therefore, while the time complexity over an already enumerated set of paths is polynomial in $h$, i.e., $\mathcal{O}(h|\mathcal{P}_{\beta}|)$, the time complexity of generating all valid paths is exponential in $h$, which is a one-time process for a given $h$, $\beta$, and $r$. However, this represents a worst-case scenario in networks where the hop length $h$ is very large. The time complexity of calculating \eqref{eq:cj_beta_expression1} is $\mathcal{O}(j|\mathcal{P}_{\beta}|C_{1})$, if binomial computation is needed, whereas with precomputed binomial values it is $\mathcal{O}(|\mathcal{P}_{\beta}|C_{1})$. In Table \ref{T:falsepositiveCalculation}, there are $(|\Delta|-1)h$ columns, therefore the time complexity of calculating whole table is $\mathcal{O}(jh|\mathcal{P}_{\beta}|C_{1}(|\Delta|-1))$. For calculating the optimal value, we evaluate \eqref{eqn:objective5} over $[1,k_2]$ values; therefore, the time complexity of the gradient descent algorithm is $\mathcal{O}(k_2)$. In the next section, we will present experimental results to validate the effectiveness of the above-defined solution.}\looseness=-1  

\section{Experimental Results} \label{sec:results}
In this section, we present experimental results to showcase the efficacy of the proposed solution in Section \ref{sec:optimization_of_BF2}.
First, we will present numerical results on the near-optimality of our solution in solving Problem \ref{problem2}. Next, we will compare our method of embedding spatial-provenance with a standard baseline wherein the GPS coordinates are embedded in the packet by each forwarding node in multi-hop communication.\looseness=-1

\subsection{Verification of Analytical Expressions}
To verify the correctness of the analytical expressions on the false positive rate for $\beta\geq1$, we compare the results of the analytical expression in Section \ref{sec:Low_complexity_C_j_any_beta} with the false positive rate generated through simulation results. For generating the simulation results, we use a network with configuration $N=6, ~h = 5$ and $|\Delta| = 15$. The size of the Bloom filter $\textbf{BF}_1$ is chosen such that it gives rise to negligible false positive rates, thereby satisfying the condition  $\Pr(E_{fp}^{(1,2)}) \approx \Pr(E_{fp}^{(2)})$. Henceforth, we denote the probability of false positive event $\Pr(E_{fp}^{(2)})$ as $P_{fp}$. For embedding spatial-provenance, $\textbf{BF}_2$ of size $m_2=100$ bits is used. To obtain the simulation results, $10^6$ packets are sent from the source node to the RSU via a multi-hop network for each value of $k_2$, where $k_2$ varies from $2\leq k_2 \leq m_2$. The corresponding $P_{fp}$ values are plotted as a function of $h$ in Fig \ref{fig:h_5_delta_15_beta_2}. From Fig. \ref{fig:h_5_delta_15_beta_2}, we observe that the curve of the analytical expression acts as a good approximation to that generated via simulation results, and the minima of both the plots coincide. Similar experiments were conducted by changing the network parameters $\beta$ and $m_2$, the results of which are depicted in Fig. \ref{fig:h_5_delta_15_beta_10}, Fig. \ref{fig:h_5_delta_15_beta_15} and Fig. \ref{fig:h_8_delta_15_beta_2}.\looseness=-1 
% We observe that the results generated through the analytical expressions discussed in Section \ref{sec:Low_complexity_C_j_any_beta} are close approximations of the simulation results. Note that the minima of the curves give the optimal number of Hash functions where the probability of false positive is minimum.
\begin{figure}[ht!]
\centering
\begin{subfigure}{0.49\columnwidth} 
\caption{}
\label{fig:h_5_delta_15_beta_2}
    \includegraphics[trim={0.3cm 0 1.5cm 1cm},clip,scale=0.3, width =\textwidth]{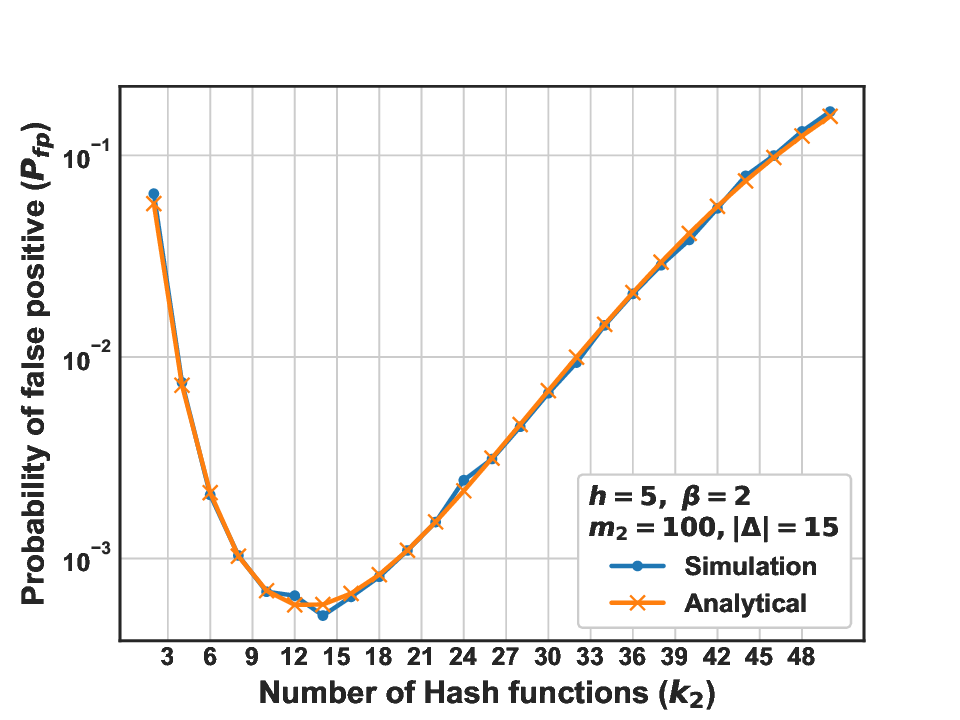}   
\end{subfigure}
% \hfill
\begin{subfigure}{0.49\columnwidth}
\caption{}
\label{fig:h_5_delta_15_beta_10}
    \includegraphics[trim={0.3cm 0 1.5cm 1cm},clip,scale=0.3, width =\textwidth]{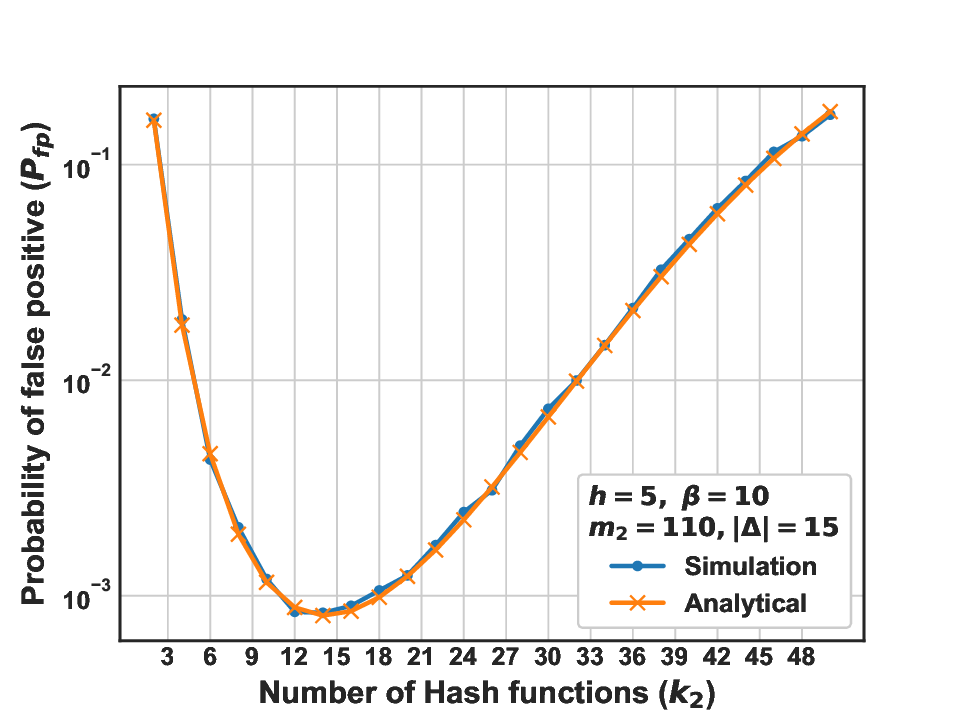}
    
\end{subfigure} 
\begin{subfigure}{0.49\columnwidth}
\caption{}
\label{fig:h_5_delta_15_beta_15}
    \includegraphics[trim={0.3cm 0 1.5cm 1cm},clip,scale=0.3, width =\textwidth]{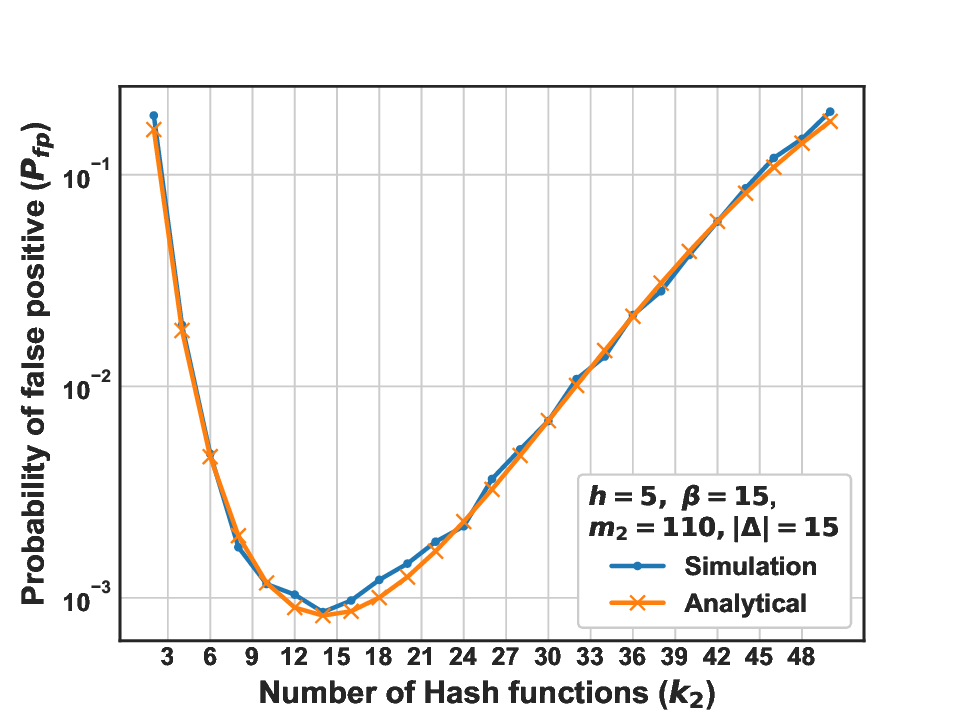}
     
\end{subfigure}  
% \hfill 
\begin{subfigure}{0.49\columnwidth}

\caption{}
\label{fig:h_8_delta_15_beta_2}
    \includegraphics[trim={0.3cm 0 1.5cm 1cm},clip,scale=0.3, width =\textwidth]{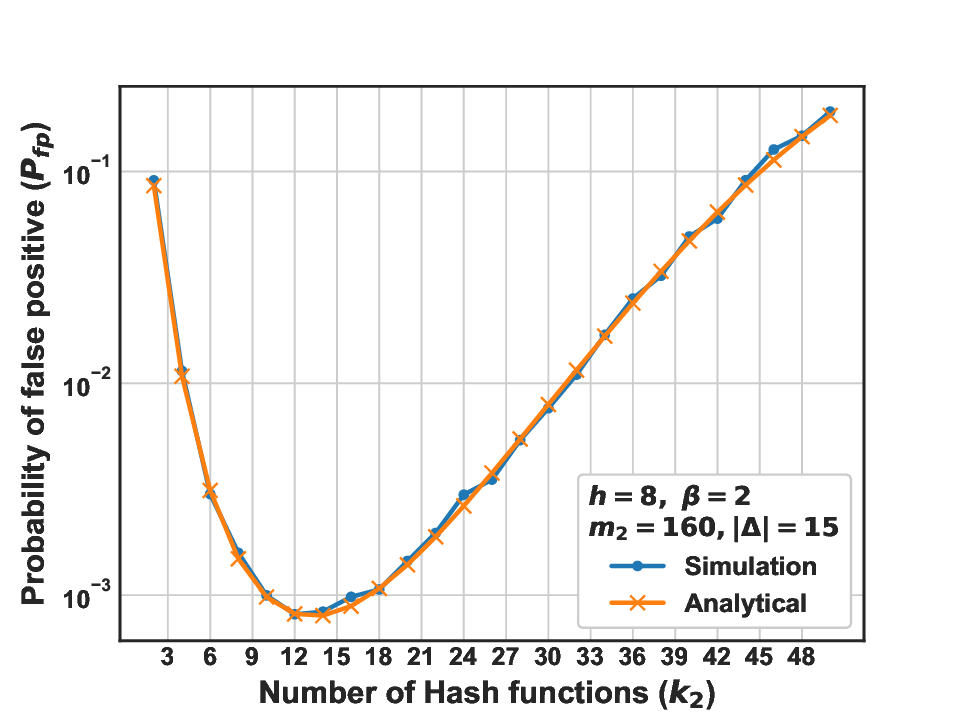}
    
\end{subfigure} 
\caption{Comparison between the analytical bound and the simulation results for the network setting of (a) $N=6,~\beta =2,~m_2=100, ~|\Delta| =15 $, (b) $N=6,~\beta =10,~m_2=110, ~|\Delta| =15 $, (c) $N=6,~\beta =15,~m_2=110, ~|\Delta| =15 $, (d) $N=9,~\beta =2,~m_2=160, ~|\Delta| =15$.\looseness=-1} 
% \vspace{0cm}
\end{figure}
\begin{figure}
\vspace{-0.5cm}
    \centering
    \includegraphics[trim={0cm 0.0cm 1.5cm 1.0cm},clip,scale=0.45]{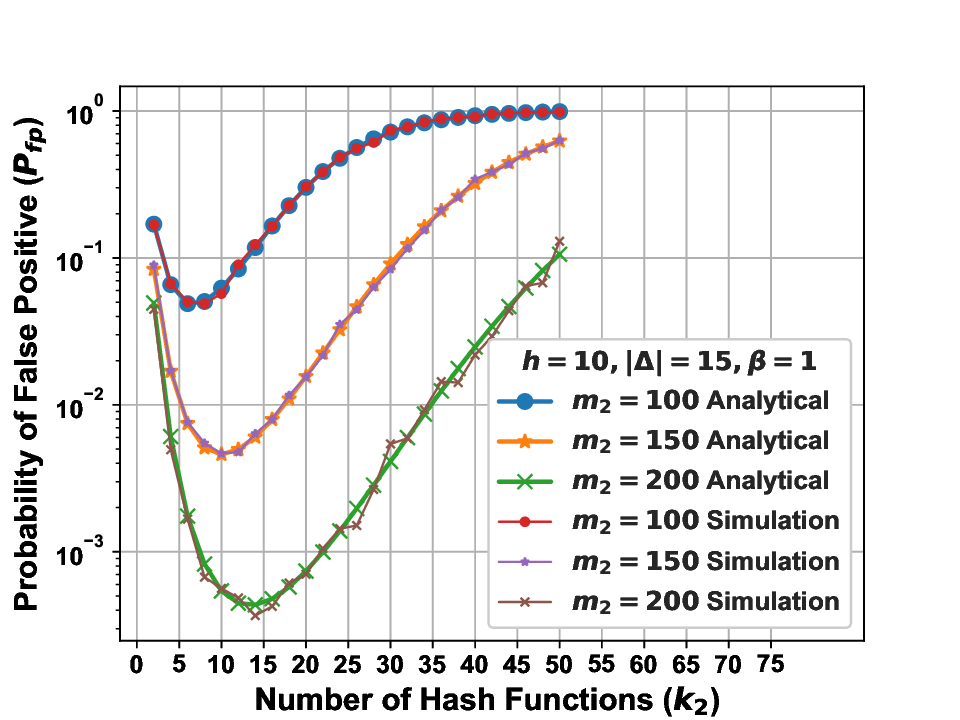}
    \caption{ Comparison of $P_{fp}$ using the analytical bound and the simulation results with varying values of $m_2$ on a network of $N = 11, ~\beta = 1$ with $|\Delta| = 15$ and $h=10$.\looseness=-1}
    \label{fig:m-varying}
    \vspace{-0.6cm}
\end{figure}

\begin{figure}
    \centering
    \includegraphics[trim={0cm 0.0cm 1.5cm 1.0cm},clip,scale=0.4]{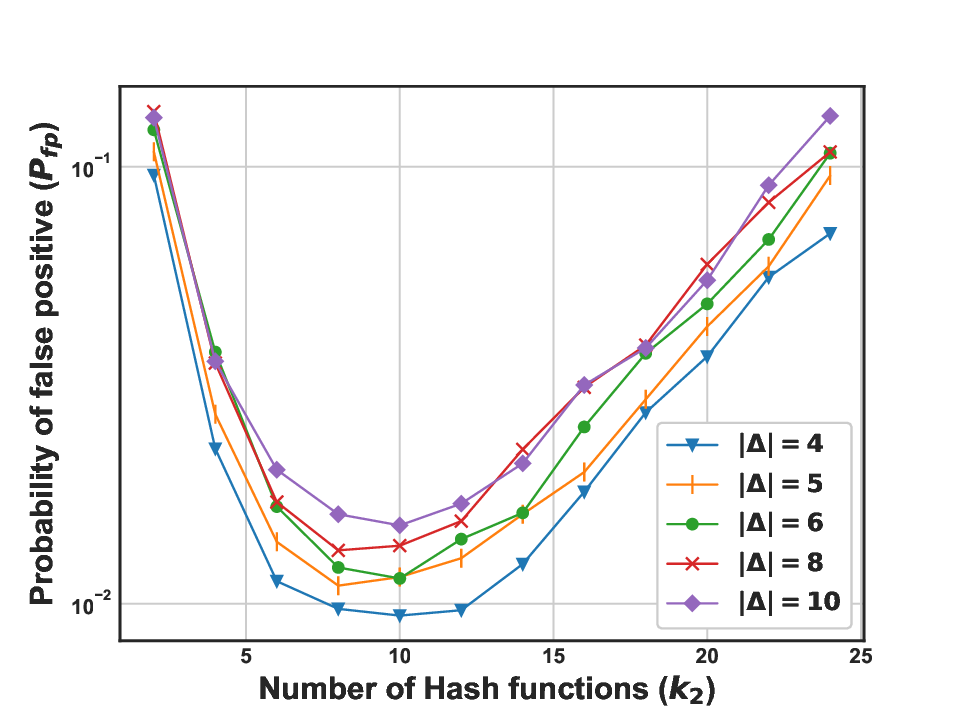}
    \caption{Comparing $P_{fp}$ of simulated results for different values of $|\Delta|$, where the underlying network parameters are $N = 16$, $m_2 = 200$ bits, $\beta = 1$, and $h = 15$.\looseness=-1}
    \label{fig:Delta_varying}
    \vspace{-0.25cm}
\end{figure}
  
 % To verify the correctness of the analytical expressions on the false positive rate for $\beta=1$, we compare the results of the analytical expression in Section \ref{sec:Low_complexity_C_j_any_beta} with the false positive rate generated through simulation results.
 To study the behaviour of $P_{fp}$ with respect to $m_2$, 
 we simulated a network with $N = 11$, $h = 10$ and $|\Delta| = 15$. With optimized $\textbf{BF}_1$, $m_2=200$ bits is used for embedding the spatial-provenance, and the number of Hash functions varies from $2 \leq k_2 \leq m_2$. For the given network parameters, we compare the analytical expression with the simulation results in Fig. \ref{fig:m-varying}. From Fig. \ref{fig:m-varying}, we observe that the curve of the analytical expression acts as a good approximation to the simulation results, and the minima of both plots coincide. 
 We present similar results by varying the value of $m_2$, where we observe that $P_{fp}$ decreases as the size of $m_2$ increases. This implies that with a larger packet size, localization accuracy at the RSU will improve subject to a given choice of $|\Delta|$ as agreed by the nodes. For the network configuration, where $N=16$, $h=15$, $m_2 =200$ bits and $\beta =1$, Fig. \ref{fig:Delta_varying} shows that if the RSU needs to learn a higher resolution of localization for a fixed  size of Bloom filter, the accuracy of localization reduces. Thus, the only way to learn a higher resolution of location with high accuracy is to increase the Bloom filter size.\looseness=-1

\begin{figure}[!tbp]
  \begin{subfigure}[b]{0.49\columnwidth}
  \caption{
     % Analysis of the relationship between $P_{fp}$ and $|\Delta|$ for a network configuration where $h = 12$ and $m_2 = 200$ bits.\looseness=-1
     }
    \includegraphics[trim={1cm 0.2cm 1.5cm 0.9cm},clip,scale=0.32]{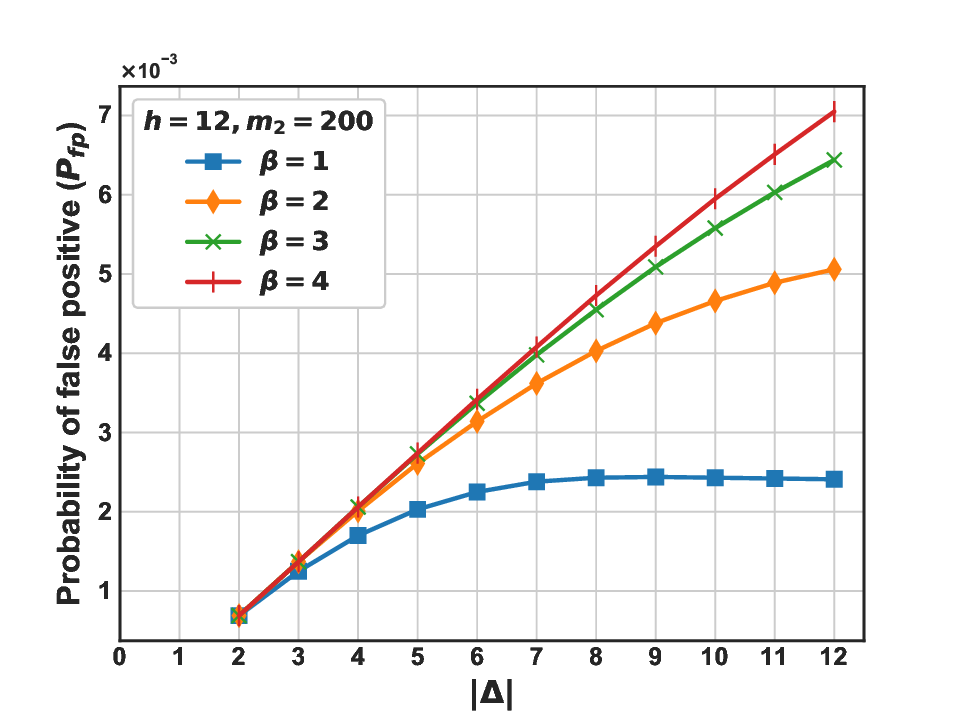}
     
     \label{fig:Tradeoff_delta_Vs_pfp}
  \end{subfigure}
  \hfill
  \begin{subfigure}[b]{0.49\columnwidth}
  \caption{
    % Analyzing the relationship between the required size of the Bloom filter and $|\Delta|$ for a network setup with $h = 12$ and a fixed false positive rate of $P_{fp} = 10^{-4}$. The magnitude represented on the y-axis is measured in bits.\looseness=-1\looseness=-1
    }
    \includegraphics[trim={0.5cm 0.2cm 1.5cm 1cm},clip,scale=0.32]{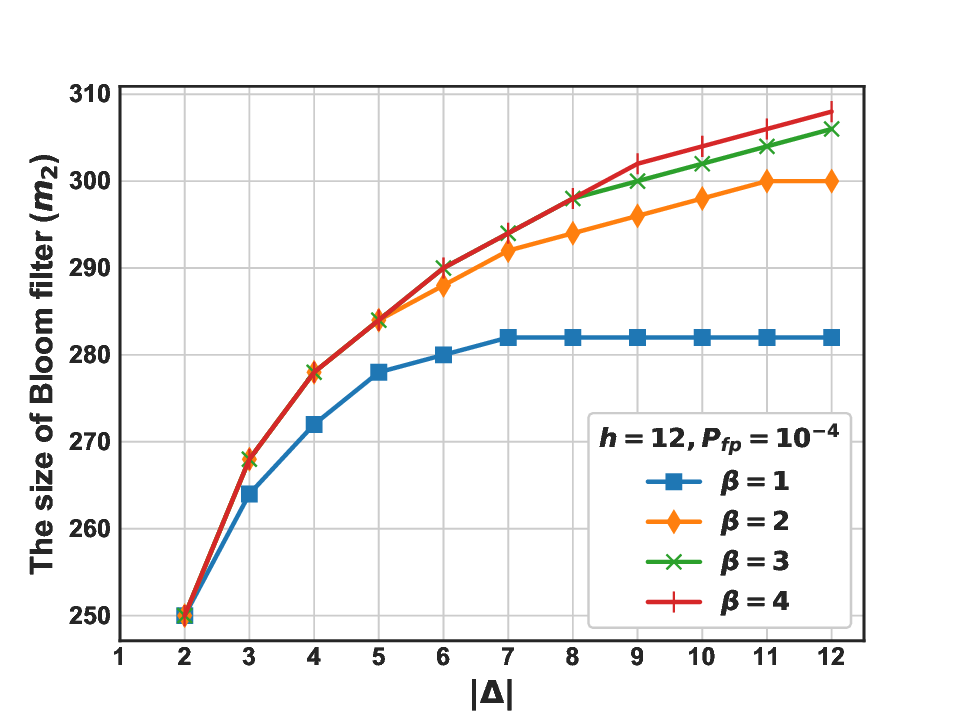}
    
    \label{fig:Tradeoff_Delta_Vs_m2}
  \end{subfigure}
  \caption{(a) Analysis of the relationship between $P_{fp}$ and $|\Delta|$ for a network configuration where $h = 12$ and $m_2 = 200$ bits. (b) Analysis of the relationship between the required size of the Bloom filter and $|\Delta|$ for a network setup with $h = 12$ and a fixed false positive rate of $P_{fp} = 10^{-4}$. The magnitude represented on the y-axis is measured in bits.\looseness=-1}
  \vspace{-0.35cm}
\end{figure}

Thereafter, we study the impact of privacy on the false positive rate as a function of $\beta$. For this study, we fix the size of Bloom filter $m_2=200$ bits for a fixed hop network of $h=12$ where $N=13$. For a given $\beta$ and fixed $|\Delta|$, we obtain $P_{fp}$ resulting from the traversal of $10^6$ packets from the source node towards the RSU via a multi-hop network. We plot $P_{fp}$ as a function of $|\Delta|$, where $2\leq |\Delta|\leq 12$, for a given $\beta$ in Fig. \ref{fig:Tradeoff_delta_Vs_pfp}. 
We repeated the same experiment with different values of $\beta$.  When $13$ nodes are distributed across $|\Delta| = 2$ segments, the spatial privacy of nodes is high, whereas when the nodes are distributed across $|\Delta| = 12$ segments, the spatial privacy of nodes is low. For a given $\beta$, which depicts the transmission power of the nodes, i.e., how far a node can transmit, from Fig. \ref{fig:Tradeoff_delta_Vs_pfp}, we observe that $P_{fp}$ increases with $|\Delta|$. An increase in $P_{fp}$ with an increase in $|\Delta|$ is the effect of the number of node-location pairs increasing with increasing $|\Delta|$. In other words, for a fixed-size Bloom filter and for a given number of nodes $N$, if we relax the localization privacy, then we will lose the reliability of recovery, which is depicted by an increase in false positives.\looseness=-1 

We also study the trade-off for the required size of the Bloom filter with respect to the number of divisions of segments. For a fixed false positive rate, $P_{fp} = 10^{-4}$, for both the Bloom filters and fixed hop length $h=12$ with $N=13$, we vary the number of segments from $|\Delta| = 2$ to $|\Delta| =12$ and plotted the required size of $m_2$ as a function of $|\Delta|$ in Fig. \ref{fig:Tradeoff_Delta_Vs_m2}. From Fig. \ref{fig:Tradeoff_Delta_Vs_m2}, we observed that for higher resolution of spatial location, a higher number of Bloom filter bits are required to maintain the $P_{fp}=10^{-4}$.\looseness=-1 

In the next section, we will compare our proposed protocol with the baseline method of sharing the GPS coordinates.\looseness=-1
\vspace{-0.25cm}
\subsection{Benefits of using Bloom Filter}
To assess the benefits of using Bloom filters for spatial-provenance, we empirically compare the communication overhead in two multi-hop forwarding scenarios: (i) nodes share encrypted GPS coordinates with the RSU, and (ii) nodes use Bloom filters to convey their segment ID using our proposed protocol.\looseness=-1

In the former method, for sharing the encrypted GPS location, we assume that the nodes append their node ID with their GPS coordinates at the end of the packet. 
To minimize the amount of information shared between nodes, we exclusively transmit the latitude and longitude coordinates at the coarse level, disregarding additional GPS parameters such as time and elevation. This results in a cumulative seven bytes of data, with one byte allocated for the node ID and six bytes for the GPS information. In order to provide confidentiality, we encrypt this data using AES-128, which increases the data length from seven to sixteen bytes. In the same manner, each forwarding node appends its sixteen bytes of encrypted location information at the end of the packet as it makes its way towards the RSU. It is assumed that the pair-wise AES keys are shared between each node and the RSU. When the packet is received at the RSU, it will recover the ID of the forwarding nodes and their corresponding GPS coordinates by successively decrypting the tail of the packet. We conducted an experiment to implement this method on a network where the number of nodes $N$ is equal to the hop length, i.e., $h=N-1$. By varying the number of nodes from four to fifteen, we collected the data on the average packet size as a function of the hop-length $h$ and plotted them in Fig. \ref{fig:Baseline_Size}. In this context, since the packet size is different at each hop, the average packet size is calculated by dividing the sum of the packet sizes at all hops by the hop length. To perform the Bloom filter experiment, we partition the RSU's coverage area into ten segments with $|\Delta| =10$ and build a multi-hop network with the number of hops as $h=N-1$ with $\beta=1$.  In this experiment, nodes share their segment information using our method of embedding spatial-provenance. The false positive rate, $P_{fp}$, for the experiment, is set at $10^{-4}$ for both the Bloom filters $\textbf{BF}_1$ and $\textbf{BF}_2$. To obtain the desired false positive rate, we use the optimal size of both the Bloom filters and the related optimal number of the Hash functions. We plot the sum of the sizes of both Bloom filters, $m_1+m_2$, as a function of $h$ in Fig. \ref{fig:Baseline_Size}.  
As shown in Fig. \ref{fig:Baseline_Size}, an increase in hop-length results in greater space advantages for the Bloom filter-based scheme as compared to the scheme of encrypted GPS.\looseness=-1 

\begin{figure}[!tbp]
  \begin{subfigure}[b]{0.49\columnwidth}
  \caption{}
    \includegraphics[trim={0.5cm 0 1.5cm 1cm},clip,scale=0.31]{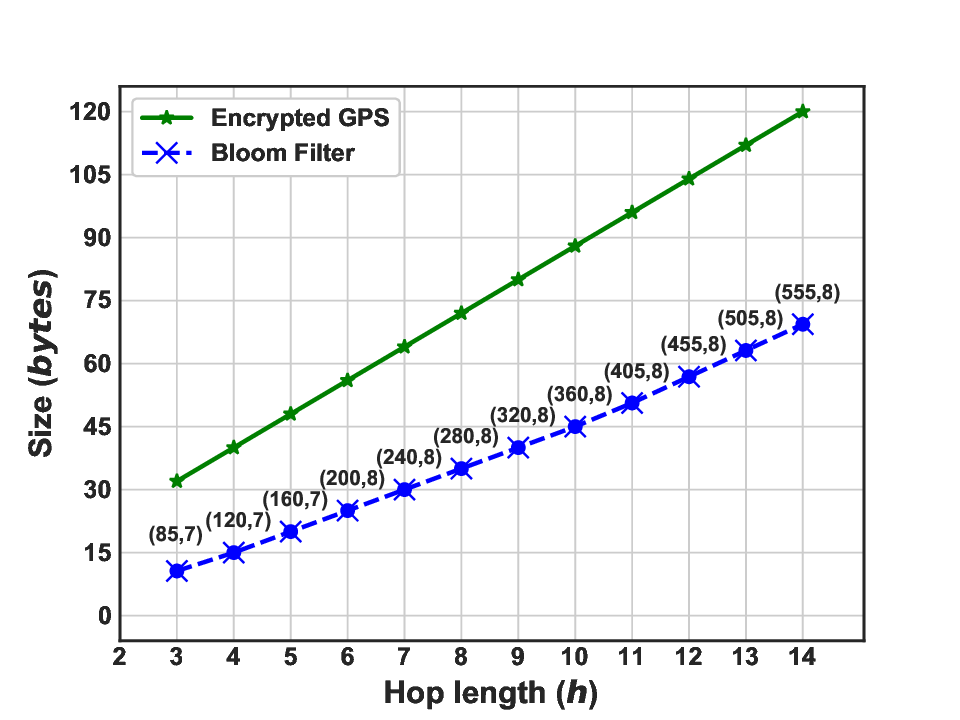}
     
     \label{fig:Baseline_Size}
  \end{subfigure}
  \hfill
  \begin{subfigure}[b]{0.49\columnwidth}
  \caption{
    % Comparing end-to-end delay of Bloom filter $\textbf{BF}_2$ with end-to-end delay of compressed Bloom filter that uses RAKE compression with $P_{fp} = 10^{-4}$.
    % % Comparison of end-to-end delay between the compressed Bloom filter provenance sharing and the uncompressed Bloom filter-based provenance sharing for $\beta = 1$. Results are generated for a false positive of $10^{-4}$.\looseness=-1
    }
    \includegraphics[trim={0.5cm 0 1.5cm 0.9cm},clip,scale=0.31]{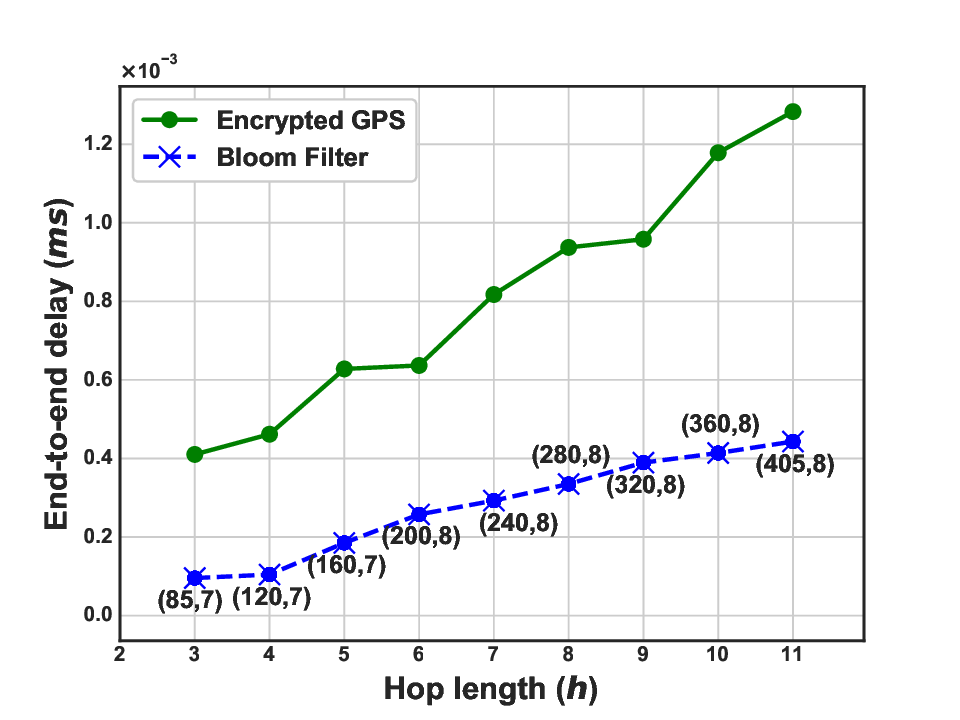}
    
    \label{fig:Baseline_time}
  \end{subfigure}
  \caption{Comparison of parameters between sharing encrypted location and sharing location using Bloom filter for $\beta = 1$ and $|\Delta| = 10$. For the Bloom filter, the results are generated for $P_{fp}=10^{-4}$. The graph displays the corresponding optimized values of $m_1+m_2$ and $k_2$ as coordinates. (a) Comparison of space-complexity (b) Comparison of end-to-end delay.\looseness=-1 }
  \vspace{-0.4cm}
\end{figure}

To capture the time-complexity overhead, we conduct an experiment using the same network setting as defined above. For time-complexity overhead, we measured end-to-end time delay, which is the sum of the time taken by each node in the forwarding path to embed its spatial-provenance using the GPS method and our proposed protocol of embedding spatial-provenance. We calculated the average end-to-end time delay over $10^6$ packets sent from the source node to the RSU for each hop-length $h$ and plotted the average end-to-end delay as a function of $h$ in Fig. \ref{fig:Baseline_time}. We observed that Bloom filters offer benefits in terms of time-complexity over the GPS-based scheme.\looseness=-1

Overall, for a smaller value of $N$, our suggested method necessitates a significantly reduced packet size; however, a large packet size may be necessary for a large $N$, in order to drive down the false-positive rate of $\textbf{BF}_2$ to a small value. However, in terms of the benefits in time-complexity, our proposed method is still better than the GPS-based scheme due to the utilisation of Hash functions.\looseness=-1 

\section{Testbed Results and Validation Using Compression Algorithms}\label{sec:testbed_results_practical_aspects}
To study the feasibility of the spatial-provenance algorithms on the state-of-the-art wireless protocols, we have implemented spatial-provenance algorithms on the ZigBee and LoRa protocols. Since there is no practical implementation of a spatial-provenance framework for existing wireless networks, no dedicated bits are allocated for the same in their packet structure. Therefore, in our prototype, we use a portion of the payload to carry the Bloom filter for spatial-provenance. While implementing our proposed protocols, we discovered that ZigBee and LoRa protocols have limited payload size. Therefore, we need to look for the possibilities of compression for spatial-provenance. In the next section, we study how to compress the provenance to suit payload constraints.\looseness=-1 
\vspace{-0.25cm}
\subsection{Motivation for Compression} \label{sec:Compression_RAKE}
Recall that, we use CLBF to embed the provenance information (refer to Section \ref{sec:Bloom Filters}). Suppose that $m_1 = 100$ bits and $m_2 = 100$ bits, and the corresponding $k_1^* = 6$ and $k_2^*=9$, respectively. When the source node embeds its provenance in Bloom filters, the maximum number of 1’s that can be lit in Bloom filters is $k_1^* + k_2^* = 15$, the rest of the bits will be 0’s, i.e., atleast 185 bits will be 0’s. Therefore, we can potentially exploit this redundancy and send the packet in a compressed format, which will reduce the payload size of the packet. Intermediate nodes in the forwarding process can decompress the packet, embed its information, and then compress it again.
Assuming no overlap in bit positions, the node at the second hop position may lit up to $30$ bits, still leaving at least $170$ bits as 0's. Therefore, compression can be repeatedly applied across forwarding nodes to exploit sparsity and reduce provenance overhead.\looseness=-1

For accurate recovery of the Bloom filter, we focus on lossless compression algorithms, regarding the choice of the compression algorithm. Among many available compression algorithms, such as Huffman encoding \cite{huffman}, run length encoding, Arithmetic encoding, and Compressed Sparse Representation, we found that RAKE  \cite{rake} compression is most suitable for compressing small binary packets and it is easy to implement and has a minimal dictionary size.\looseness=-1 

To study the benefits of RAKE compression in our method for a given hop-length $h$, we use the optimal size of $\textbf{BF}_2$ and an optimal number of Hash functions for $P_{fp} = 10^{-4}$. The source node embeds its segment ID in $\textbf{BF}_2$ as per our proposed protocol, compresses it using RAKE compression and forwards it towards the next node. On reception of a packet at the next node, it will decompress $\textbf{BF}_2$ and embed its segment ID in $\textbf{BF}_2$, compress it and forward it towards the next node. Similarly, decompression, the embedding of segment ID and compression operations will be performed at each forwarding node. We measure the average provenance size, which is calculated as the ratio of the sum of the size of compressed provenance at each node to the number of hops. With $10^{6}$ packets sent from the source node towards the RSU for each $h$, we plot the average compressed size of $\textbf{BF}_2$ as a function of $h$ in Fig. \ref{fig:compression_benefit_size}. We observed that in a space-constrained scenario, RAKE compression will be useful. In the same experiment, we also measure the end-to-end time delay, which is plotted in Fig. \ref{fig:time_compressVsUncompressed}. We observe that there is a tradeoff with respect to time in using RAKE compression as extra time taken in compression and decompression at each forwarding node.\looseness=-1 

\begin{figure}[!tbp]
  \begin{subfigure}[b]{0.49\columnwidth}
  \caption{}
    \includegraphics[trim={0.5cm 0.2cm 1.5cm 1.2cm},clip,scale=0.31]{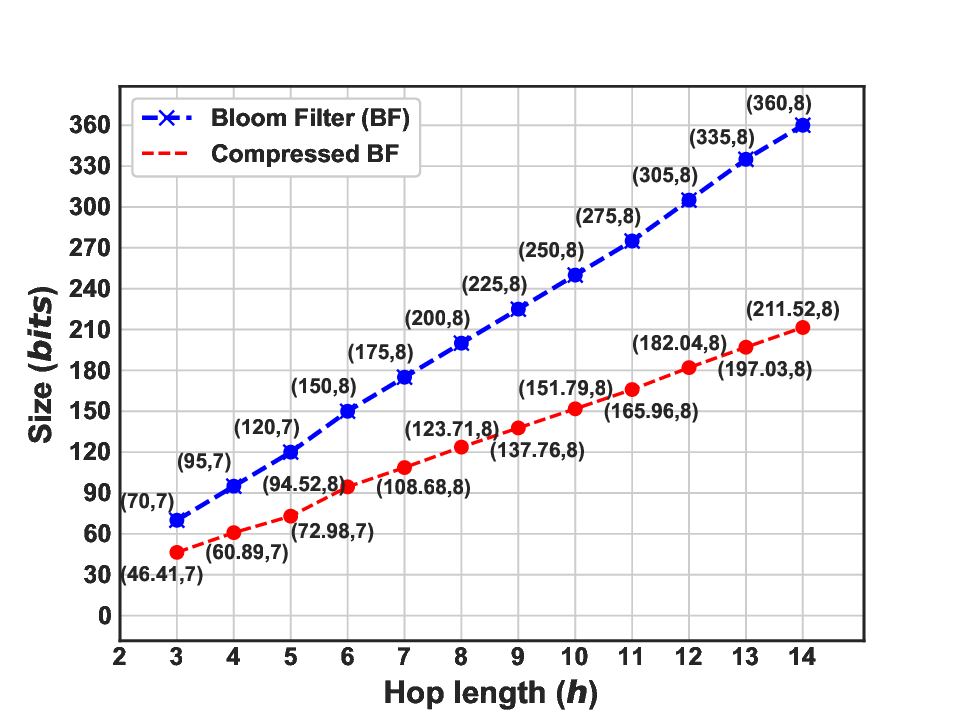}
     
    \label{fig:compression_benefit_size}
  \end{subfigure}
  \hfill
  \begin{subfigure}[b]{0.49\columnwidth}
  \caption{
    }
    \includegraphics[trim={0.4cm 0 1.5cm 0.9cm},clip,scale=0.31]{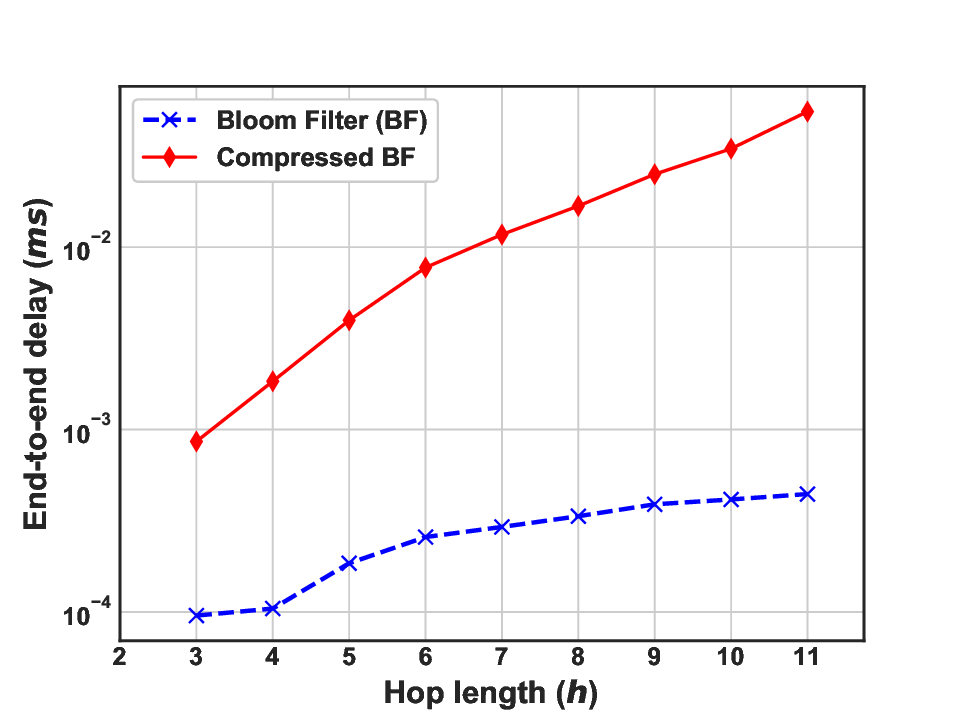}
    
    \label{fig:time_compressVsUncompressed}
  \end{subfigure}
  \caption{Comparing the parameters of $\textbf{BF}_2$ with those of a compressed Bloom filter that uses RAKE compression with $P_{fp} = 10^{-4}$. (a) Comparison of the required size of Bloom filter. The first coordinate on the plots depicts the $\textbf{BF}_2$ size, and the second coordinate shows the optimal number of Hash functions. (b) Comparing end-to-end delay.\looseness=-1 }
  \vspace{-0.3cm}
\end{figure}
\vspace{-0.25cm}
\subsection{Experimental Configuration for Spatial-Provenance}
The testbed system, depicted in Figure \ref{fig:XBee}, comprises XBee S2C \cite{XBee_S2C} devices that operate on the ZigBee protocol inside the ISM band. 
We utilize Semtech's LoRa modules \cite{LoRa} for long-distance communication, which operates river ISM bands. For computational needs, we utilize Raspberry Pis and high-performance computing devices such as laptops. To showcase a stationary vehicular network with limited-range communication, we employ a combination of Raspberry Pi and XBee S2C devices to simulate a vehicle. However, for long-distance communication, we employ LoRa and Raspberry Pi to create a static vehicular model. In both of these scenarios, a high-performance computing device functions as the Roadside Units (RSUs) and can utilize either an XBee or LoRa technology.\looseness=-1
  
\begin{figure}[!h]
    \centering
    \includegraphics[width = 0.4\textwidth]{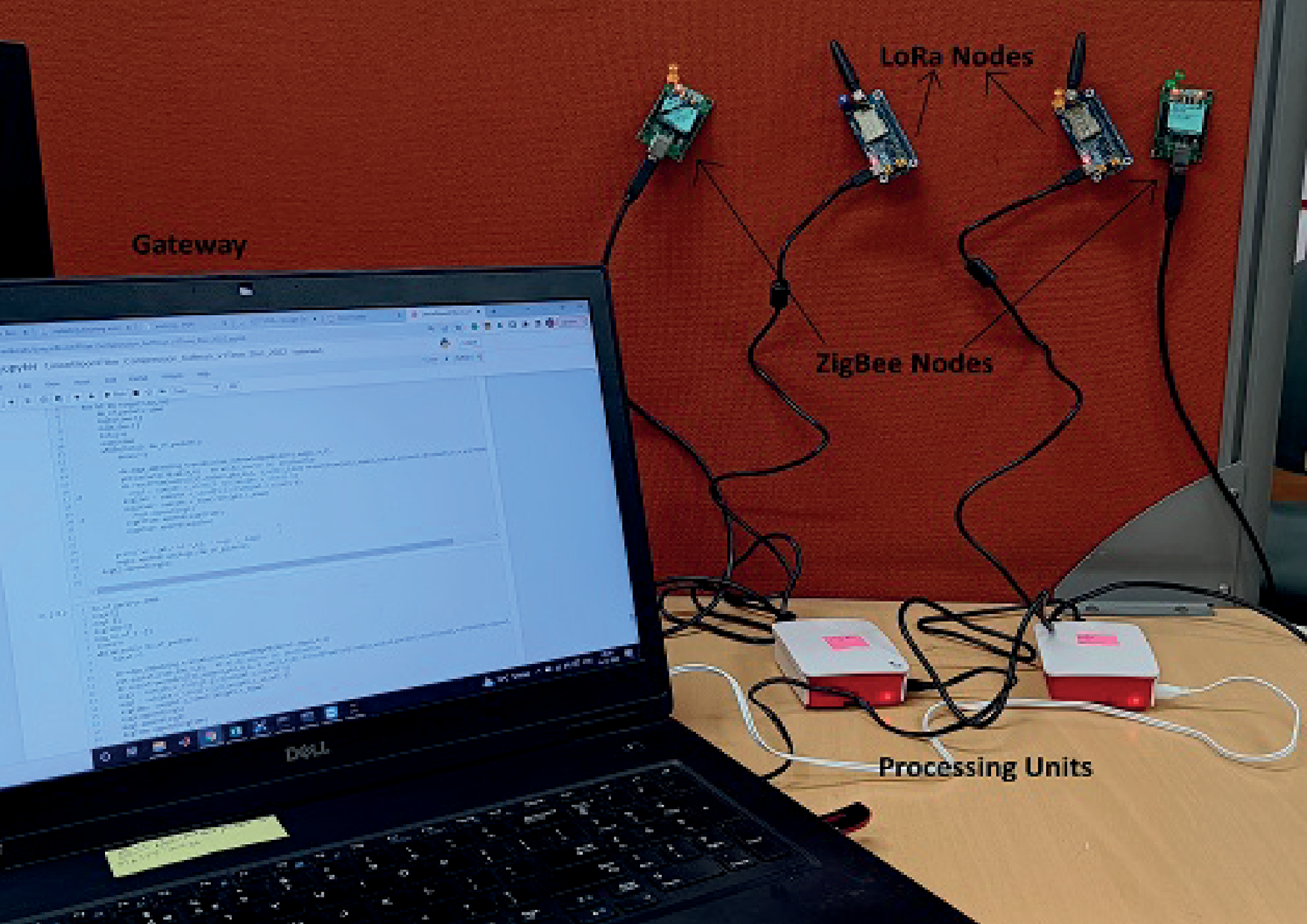}
    \setlength{\belowcaptionskip}{-10pt} 
    \caption{A testbed consisting of XBee and LoRa devices showcases a static vehicular network.\looseness=-1}
    \label{fig:XBee}
\end{figure}

\subsubsection{Hardware Setup}
To prepare the setup, we distribute the vehicular nodes across a geographical area such that any node can communicate with the RSU in a multi-hop manner. Given that vehicles are static in our setup, we hardcode their location information on them so that they can identify their segment identity upon receiving the dictionary from the RSU. For the routing protocol, AODV is used for multi-hop communication from every vehicle to the RSU. For downlink communication, the RSU communicates directly with every vehicle in a single hop.\looseness=-1

As discussed in Section \ref{sec:BF_embedding}, a source vehicle embeds its spatial-provenance into the Bloom filter and forwards it to the next vehicle in the path. Subsequently, the next vehicle repeats the process of embedding its spatial-provenance until the packet reaches the RSU. To execute these steps, we choose the underlying parameters, such as the number of Hash functions for the Bloom filter and the Bloom filter size, based on an offline optimization process explained in Section \ref{sec:optimization_of_BF2}. Furthermore, as discussed in Section \ref{sec:Compression_RAKE}, we ask each vehicle to use RAKE compression. Consequently, every vehicle that intends to embed its spatial-provenance implements a RAKE decompression algorithm on the reception of a packet. Finally, once the packet is received at the RSU, it verifies the location of the vehicles using their identities and the dictionary as explained in Section \ref{sec:recovery}. This way, every vehicle is localised at the RSU, respecting the relaxed-privacy constraints of the vehicles.\looseness=-1

\subsection{Takeaways from Our Testbed}
We performed multiple experiments on our testbed and measured the space- and time-complexity of our framework for a 5-hop network where 10 vehicles are scattered across 5 segments. The vehicles employ a Bloom filter with a capacity of 100 bits and utilize eight Hash functions to incorporate the IDs of their segments. The average sparsity of a Bloom filter is calculated by dividing the total of the number of bits that are set to 1 in the Bloom filter across all hops by the product of the number of hops and the size of the Bloom filter. Upon using the RAKE compression technique at each hop, we further compute the average provenance size, which is determined by dividing the total size of the compressed provenance at each node by the number of hops. The metrics mentioned above are presented in Table \ref{tab:compression_table} after doing ten thousand rounds of iterations for the Bloom filter sizes of 100, 125, and 150. For example, if we use a Bloom filter with 100 bits, the average sparsity will be $21\%$, and this can be compressed to $77$ bits. Hence, when utilizing XBee, which possesses a predetermined payload capacity of 255 bytes, our spatial-provenance method only requires a mere $3.9\%$ of the payload. When utilizing LoRa with a data rate of 7, which has a constant payload size of 222 bytes, our spatial-provenance method takes up $4.5\%$ of the payload. Without compression, our spatial-provenance would take up around $6\%$ of the payload in both XBee and LoRa. Although RAKE compression decreases the space used, it also results in an increase in end-to-end packet delay compared to the method without any compression algorithm.\looseness=-1

The LoRa network is employed for long-distance communication, resulting in a significantly greater coverage area. In contrast, the XBee network has a narrower coverage area because it is designed for short-distance communication. Thus, when comparing the number of segments, LoRa provides greater absolute privacy while XBee offers less. However, the normalized privacy remains equal in both technologies. In summary, our experiment demonstrates that spatial-provenance algorithms can be implemented on ZigBee and LoRa networks, provided that the RSU is capable of acquiring localization data with low-to-moderate precision. However, achieving a better level of accuracy in localization can only be accomplished by utilizing a greater proportion of the available payload capacity. Next, we present the practical aspects of \bl{the} dictionary broadcast mechanism by the RSU.\looseness=-1

\begin{table}[!h]
    \centering
    \caption{Compressed provenance size.}
    \resizebox{6cm}{!}{\begin{tabular}{|c|c|c|}
    \hline
         Provenance size  & Sparsity (Avg.)   & Avg. Provenance size  \\(bits) & $\%$ & after compression (bits)  \\
    \hline
         100 & 20.92 & 76.92 \\
         125 & 17.17 & 86.34\\
         150 & 14.58 & 94.03 \\
    \hline
    \end{tabular}}
    \label{tab:compression_table}
\end{table}

\vspace{-0.3cm}

\subsection{Practical Aspects of Implementing Spatial-Provenance} \label{sec:Dictionary}
Practical implementation of the proposed spatial-provenance protocol requires a reliable mechanism for dictionary broadcast by the RSU. As discussed in Section \ref{sec:network_model}, dictionary mappings are predetermined by an amicable solution between the localization requirement of the RSU and the privacy requirement of the vehicles. We envisage that the dictionary mappings remain static unless there is a network-wide reconfiguration, such as a change in segment granularity. The RSU broadcasts the dictionary containing GPS-to-segment ID mappings periodically on the downlink.\looseness=-1

We envision that a vehicle receives a dictionary when it enters the coverage area through a periodic broadcast by the RSU. Once a vehicle receives the dictionary, it caches it locally and uses it throughout its presence in the coverage area of the RSU. With the help of the dictionary, vehicles map their real-time GPS coordinates to the corresponding segment ID using a lightweight local computation process. Subsequently, this segment ID is used during the packet forwarding process to embed the segment ID onto $\textbf{BF}_2$. Next, we present an analysis on the choice of the periodicity of the dictionary broadcast, and its impact on our proposed protocol.\looseness=-1

Let the linear coverage area of the RSU be $L$ meters in length, which is divided into $r \in \mathbb{N}$ equal-length segments, each of length $l=\frac{L}{r}$ meters. The dictionary, of size $B$  bytes, is broadcast periodically at an interval of $\tau_b$ seconds, and the corresponding transmission delay of dictionary over a wireless link of bandwidth $R$ bits/second is $\tau_{tx}= \frac{8B}{R}$ seconds, and the propagation delay from the RSU to the farthest segment in its coverage area is $\tau_{prop} = \frac{L}{c}$ seconds, where $c$ is the signal propagation speed. Therefore, the maximum end-to-end delay over the downlink for the dictionary to reach the farthest segment in the coverage area is by $\tau_t=\tau_{tx} +\tau_{prop}$ seconds. After receiving the dictionary, the vehicle spends an additional $\tau_d$ seconds in parsing and caching it locally. Assuming vehicles entering the coverage area of the RSU at random times relative to the periodic dictionary broadcast, the waiting time for the vehicle to receive the dictionary after entering the farthest segment is modeled as a uniformly distributed random variable $\tau_l \sim \mathbf{U}(0,\tau_b)$, where $\mathbf{U}(0,\tau_b)$ denotes the uniform distribution over the interval $[0,\tau_b]$. Therefore, at the vehicle, the dictionary is available for use after $\tau_l+\tau_t+\tau_d$ seconds.\looseness=-1

We now define the conditions on $\tau_b$ to ensure that a vehicle is able to use the dictionary while the vehicle is in the segment when implementing our proposed protocol. Assuming the maximum speed of the vehicle is $v$ meters per second, the segment retention time, i.e., the time a vehicle remains in a segment, is $\tau_s=\frac{l}{v}$ seconds. To ensure that each vehicle has sufficient time within a segment to receive and cache the dictionary before it moves out of the segment $\tau_b+\tau_t+\tau_d \ll \tau_s$, must be met. Typically, $\tau_s$ is of the order of several seconds, and $\tau_t$, $\tau_d$ are in the order of milliseconds; $\tau_b$ can be chosen in a straightforward manner to satisfy the above inequality. We will now study the choice of $\tau_b$ as a function of data packet arrival rate at the vehicle.\looseness=-1

\begin{figure}[ht!]
 \centering
  \includegraphics[trim={0cm 0 0cm 0.8cm},clip,scale=0.45]{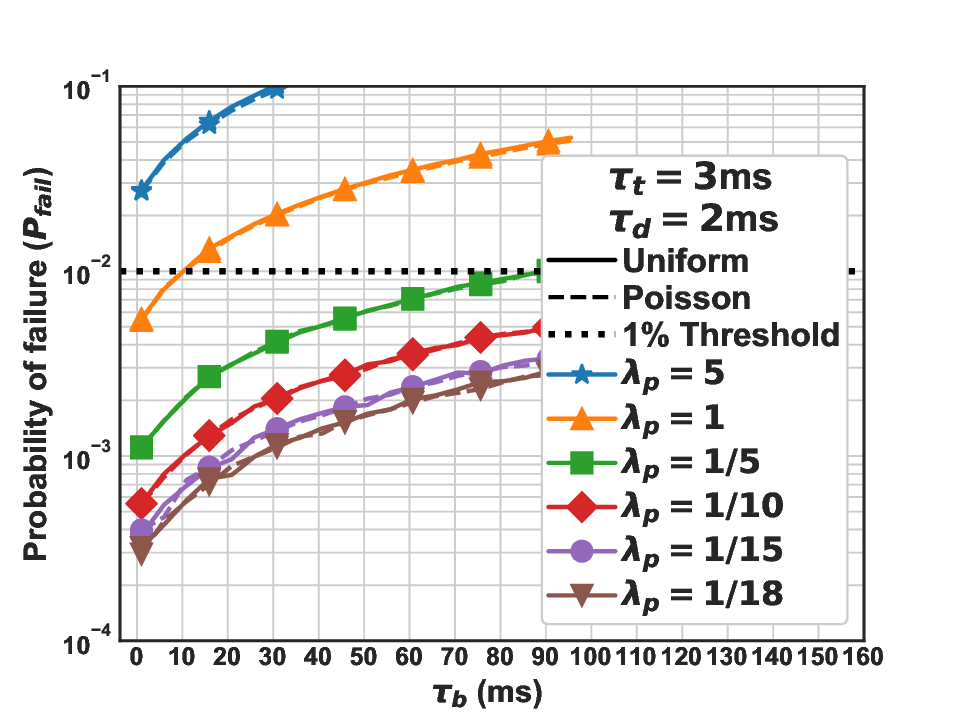}
%   \vspace{-1.3cm}
  \caption{Analysis of $P_{fail}$ with respect to $\tau_b$ for different values of packet arrival rate $\lambda_p$ (packets per second), under uniform and Poisson data packet arrival model.}
  \label{fig:P_fail_results_log_1p}
  \vspace{-0.2cm}
\end{figure}

We assume that data packets are disseminated in the network as part of an ongoing stochastic process with rate $\lambda_p$ packets per second. As the vehicle enters the coverage area at a random point in time, the time until the vehicle receives its first packet in the segment is modeled as $\tau_a$ seconds. Meanwhile, $\tau_p=\frac{1}{\lambda_p}$ seconds is the mean inter-packet arrival interval at the vehicle. A complete failure in embedding the segment ID into a packet occurs if the vehicle receives the first data packet before the dictionary is ready to be used, i.e., $\tau_l+\tau_t+\tau_d > \tau_a$. To ensure timely delivery and usability of the dictionary, i.e., a vehicle can use the dictionary for embedding its segment ID onto $\textbf{BF}_2$, the broadcast interval should satisfy a general design constraint $\tau_b + \tau_t + \tau_d < \tau_p$.
However, in realistic scenarios, both $\tau_l$ and $\tau_a$ at the vehicle are random variables and independent of each other. As a result, even if $\tau_b$ is chosen to satisfy the above inequality, there is a likelihood that a vehicle receives a data packet before it can process and use the dictionary to determine the segment ID. The probability of this event is defined as the probability of failure, which is given by:\looseness=-1
\begin{IEEEeqnarray}{rCl}
\label{eq:Probability of failure}
P_{fail}= \Pr\left(\frac{\tau_l + \tau_t + \tau_d}{\tau_a}> 1\right).
\end{IEEEeqnarray}

System designers can appropriately configure the parameters $\tau_b$ and $\tau_p$ based on $\tau_t$, $R$, and $\tau_s$ to ensure a negligible $P_{fail}$.\looseness=-1

\subsubsection{Experimental Results on Probability of Failure}
We empirically evaluate the variation of $P_{fail}$ with respect to $\tau_b$ under different assumptions about packet arrival patterns and network conditions. For the experiments, we use $\tau_t=3$ms, $\tau_d=2$ms, $\tau_b \in [1,100]$ms and $\lambda_p$ with six different values, namely, $5$, $1$, $1/5$, $1/10$, $1/15$, and $1/18$ packets per second. We assume that a vehicle enters the farthest segment in the coverage area of the RSU with the maximum speed $v=100$km per second, and the corresponding $\tau_s =18$ seconds. We simulate $\tau_a$ under two models (i) uniform distribution, where $\tau_a \sim \mathbf{U}(0, \tau_p)$, and (ii) Poisson distribution, where $\tau_a \sim \mathbf{Exp}(\lambda_p)$, such that $\mathbf{Exp}(\lambda_p)$ denotes the exponential distribution with rate $\lambda_p$.\looseness=-1

To compute $P_{fail}$ in \eqref{eq:Probability of failure}, $10^6$ samples are taken for each setting of $\tau_b$ and $\tau_p$. For each of $10^6$ samples, we compare $\tau_l+\tau_d+\tau_t$ with $\tau_a$ and count the fraction of failure events as defined in \eqref{eq:Probability of failure} for both uniform and exponential distributions. The corresponding results are plotted in Fig. \ref{fig:P_fail_results_log_1p}, where $P_{fail}$ is plotted against $\tau_b$, which \bl{depicts} that $\tau_b$ can be flexibly chosen over a wide range of values while maintaining $P_{fail}$ well below $1\%$ under varying values of $\tau_p$. These results confirm that $\tau_b$ can be chosen appropriately to ensure the practicality of the proposed spatial-provenance protocol.\looseness=-1 

Apart from the failure events captured in \eqref{eq:Probability of failure}, embedding errors may occur when a vehicle receives a packet near a segment boundary and moves into the next segment before it completes the embedding process. Specifically, this error arises if the vehicle queries its cache and embeds the segment ID of the new segment rather than the segment in which the packet was received. The likelihood of such an error depends on the speed of the vehicle and the time interval between receiving the packet and forwarding it after embedding the segment ID. In realistic scenarios, this interval is of a few milliseconds, during which the vehicle covers only a very short distance. Consequently, the probability of embedding an incorrect segment ID due to vehicle movement near the boundaries of the segment ID during this interval is negligibly small.\looseness=-1

\section{Privacy and Security Aspects of the Spatial-Provenance Protocol} \label{sec:Security_analysis}
In this section, we will first quantify the privacy on the locations of the vehicles under our relaxed-privacy model. After that, we discuss the resilience of the proposed spatial-provenance protocol against standard security threats, and subsequently explain how it helps in localizing the segment affected by a jamming attack.\looseness=-1
% \vspace{-0.25cm}
\subsection{Quantification of Relaxed-Privacy}\label{sec:quantification of relaxed privacy}
In this section, we quantify the privacy of the location of the vehicle offered by our proposed protocol. Let the precision of the GPS coordinates be $x$ square meters and the coverage area of the RSU be $y$ square meters. Then the coverage area of the RSU can be partitioned into $M \triangleq \frac{y}{x}$ Voronoi regions, and as a result, the RSU can resolve an uncertainty of $\log_2{(M)}$ bits (information-theoretically) each time a vehicle shares its GPS coordinates. However, when using the proposed relaxed-privacy framework, if the coverage area of the RSU is divided into $r$ segments, for some $r \in \mathbb{Z}_{+}$, then each segment roughly comprises $\frac{M}{r}$ GPS Voronoi regions. Therefore, with the proposed protocol, the RSU can resolve an uncertainty of $\log_2{(r)}$ bits, with a residual uncertainty of $\log_2{(M/r)}$ bits. Thus, relative to using the GPS, our notion of relaxed-privacy can be quantified as $\log_2{(M/r)}$ bits with respect to the RSU.\looseness=-1 

With reference to an external eavesdropper positioned between the RSU and the source vehicle, suppose that the eavesdropper has complete knowledge of the underlying segmentation and the transmission range of the vehicles. Then, when an eavesdropper captures a packet, its uncertainty with respect to the vehicle's location is $\log_2{(2\beta M/r)}$ bits, where $\beta$ captures the number the segments that a vehicle can transmit across. This expression follows because the eavesdropper can only be certain that the packet was either sent or received within a segment that is no more than $\beta$ segments away, and the multiplication factor of 2 accounts for the fact that the packet can be received by the eavesdropper from either direction. Regarding the vehicles responsible for forwarding the packets, the location privacy of the $i$-th vehicle of a $h$-hop path, when the packet is received at the $j$-th vehicle of the $h$-hop path, is $\log_2{(|j-i|\beta M/r)}$, where $i \neq j$.\looseness=-1

\subsection{Security Analysis}\label{sec:security}
In this section, we will discuss the resilience of our proposed methods against the packet drop attacks and the man-in-the-middle (MitM) attacks.

{\bf Packet-Drop Attacks :}
Our proposed protocol is resilient to packet-drop attacks under the assumption that it uses some of the advanced variants of AODV (refer Section \ref{sec:Routing Constraints}), which inherently handle such threats \cite{PDA_AODV, PDA_AODV_2}. Moreover, our proposed spatial-provenance protocol operates independently of the routing protocols, and in case of apprehensions of packet-drop attacks in a given deployment, our proposed spatial-provenance protocol can be layered on advanced routing protocols \cite{Hu2005,PDA4} having strong packet-drop defense mechanisms.\looseness=-1

{\bf MitM Attacks:}
In our spatial-provenance framework, each forwarding node embeds its segment ID onto $\textbf{BF}_2$ using $k_2$ Hash functions (refer to Section \ref{sec:BF_embedding}). To execute a MitM attack, an attacker may flip bits in the Bloom filter to induce a false positive event. However, such tampering is easily detectable at the RSU, as the RSU possesses the secret keys of all nodes, and it also independently verifies each contribution by recomputing Hash values during provenance recovery.\looseness=-1 

For executing a successful MitM attack in the context of this work, an adversary has to insert the information in CLBF in such a way that the RSU recovers a single path from the received CLBF, and the recovered path is a false path, other than the true path, which simultaneously satisfies the routing constraints mentioned in Section \ref{sec:Routing Constraints}. We now examine the resilience of our protocol against both external and internal MitM attacks.\looseness=-1

\noindent 1) {\bf External MitM Attacks:}
Our spatial provenance protocol uses CLBF, comprising $\textbf{BF}_1$ and $\textbf{BF}_2$. A successful MitM attack requires consistent embedding in both filters. However, $\textbf{BF}_1$ employs edge embedding with node-specific secret keys (refer to Section \ref{sec:BF_ingredients}), which are also used for embedding segment IDs in $\textbf{BF}_2$. Since external adversaries lack these keys, they cannot embed valid information in both filters. Such attacks are well studied in the context of Bloom filter-based provenance techniques in \cite[Section 6.1]{amogh}, which demonstrates that the probability of a successful external MitM attack is negligibly small. Since our proposed protocol also employs Bloom filter-based provenance mechanisms, we expect a similarly negligible probability of success for such attacks on our protocol.\looseness=-1
    
\noindent 2) {\bf Internal MitM Attacks:}
Assume an internal malicious node attempts to inject false information into the forwarding packet. The internal malicious node has its own secret key and can embed a false segment ID during its forwarding operation, which is localized to its own contribution in the overall provenance. However, with very high probability, an internal node cannot alter the segment ID embedded by other nodes in the Bloom filter, as the entries in the Bloom filter are generated using the secret keys, which are unique to each legitimate node. Such internal MitM attacks are well studied in \cite[Section 6.1]{amogh}, which demonstrates that the probability of success of such an attack is negligibly small for Bloom filter-based provenance schemes. Similarly, the usage of Bloom filter-based provenance in our proposed protocol makes it robust to internal MitM attacks with high probability.\looseness=-1

\subsection{Security Use-Case for Spatial-Provenance}

\begin{figure*}
    \centering
    \includegraphics[trim={0.0cm 0.0cm 0cm 0cm},clip,scale=0.5]{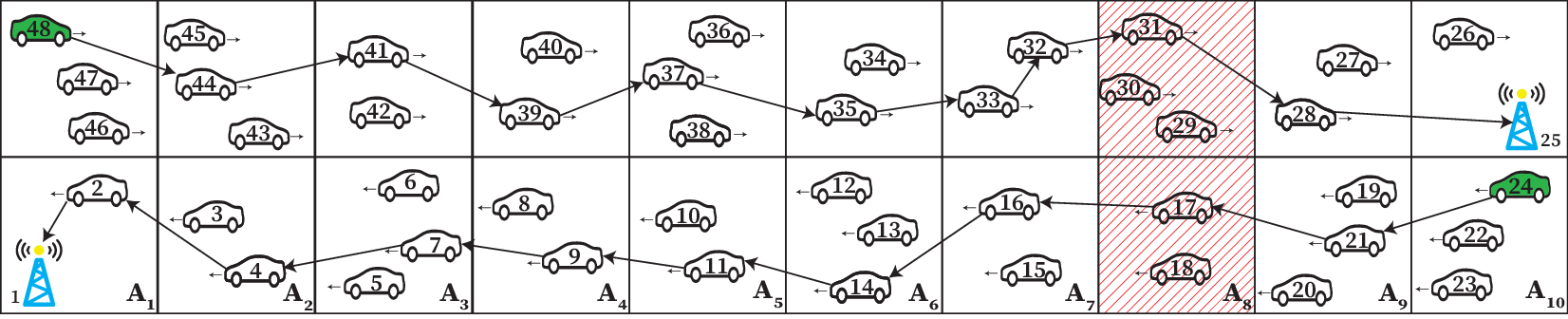}
    \caption{The figure illustrates a security use-case of spatial-provenance in a dual-road situation with 48 nodes dispersed throughout 10 segments. On one side of the road, 24 vehicles are moving towards the RSU, which is in segment $A_1$. On the other side of the road, vehicles are moving towards the RSU, which is located in segment $A_{10}$. A low-power jammer is present in Segment $A_8$.\looseness=-1}
    \label{fig:Network_threat_double road_v1}
    \vspace{-0.4cm}
\end{figure*}

In this section, first, we will highlight a practical use case in vehicular networks, wherein our proposed method of recovering spatial-provenance will be useful in mitigating a denial-of-service threat. Subsequently, we will explain how this use-case was implemented on our testbed.\looseness=-1

Consider a vehicular network where 23 nodes, denoted by $I_2, I_3,\ldots, I_{24}$, are distributed throughout ten 
segments along a linear road, as depicted in the lower single lane of Fig. \ref{fig:Network_threat_double road_v1}. Assume that a node located in segment $A_{10}$ establishes communication with the RSU using a multi-hop network and exclusively conveys only the path information in the packet, which refers to the metadata regarding the route followed by the packet. Under these circumstances, the RSU will solely regain the route but not the precise location of the intermediate nodes. \bl{Suppose a low-power jammer that can inject jamming energy within a segment becomes active in segment $A_8$, causing all the nodes in the segment $A_8$ to lose the ability to communicate with the RSU.} Consequently, the nodes located in $A_9$ and $A_{10}$ segments will be unable to establish communication with the RSU in a multi-hop manner. After the activation of the jammer, the RSU will no longer have visibility of the network beyond segment $A_7$. It will then infer that 8 out of the 23 nodes have either malfunctioned or relocated outside the coverage region. The RSU will have a pseudo picture of the network that its coverage area has only 16 nodes. Nevertheless, the RSU fails to detect any indication of a problem within its coverage area, let alone identify the specific place where the problem occurred. However, as the nodes from the jammed segment transition to an un-jammed segment, i.e., $A_7$, the RSU will detect that nodes $I_{18}$ and $I_{17}$ have initiated communication. The RSU will be unable to determine the precise underlying cause of the issue that is preventing certain nodes from communicating temporarily.\looseness=-1   

In the same scenario as explained above, assume that the nodes use a way to communicate both the path information and the spatial-location information to the RSU, using $\textbf{BF}_1$ and $\textbf{BF}_2$ as per our proposed method. In this case, if $I_{17}$ and $I_{18}$ do not reply for a certain time, and then start responding from segment $A_7$ after some time. Since the RSU recovers the spatial-provenance from every packet, it will suspect that there is an issue beyond segment $A_7$. The RSU will ping all the nodes in segments $A_8, ~A_9$ and $A_{10}$. However, throughout the specified response window, the RSU will not get any response from nodes located in these segments due to the presence of a jammer. Therefore, the spatial-provenance method has reduced some uncertainty in localizing the threat vector, and the RSU has a rough idea about the problem area that there is some problem beyond the segment $A_7$.\looseness=-1

If one wishes to further localize the threat segment, then one can envision a dual road, wherein there is one more line of traffic moving in the opposite direction handled by another RSU, as shown in Fig. \ref{fig:Network_threat_double road_v1}.
% The traffic is moving in both directions, as shown in Fig. \ref{fig:Network_threat_double road_v1}. 
If both the RSUs have a backhaul coordination between the two,  then they can use the information to precisely localize the threat segment. In this case, the jamming area can be precisely localized at the segment level, as the RSU on the right side will not be able to communicate with the nodes that are lying beyond the segment $A_9$, i.e., $A_1$ to $A_8$, and the RSU on the left-hand side will not be able to communicate with the nodes lying beyond the segment $A_7$ of the dual road.\looseness=-1

To demonstrate this security use-case on our testbed, we considered a network of 10 XBee devices spread across the coverage area of the RSU, comprising 10 segments, each of length 100 meters. Then, we mimicked the impact of a low-power jammer by disabling the radio functionalities of the underlying wireless devices belonging to a certain segment chosen at random. Subsequently, messages were exchanged between the devices and the RSU according to the proposed protocol in order to assist the RSU in localizing the rest of the devices. By comparing the location information of the nodes before and after the mimicked jamming attack, the RSU was able to identify the segment beyond which the jammer existed with high probability.\looseness=-1 

Building on this general case, we now consider a practical scenario where the coverage area of the jammer is known in terms of the number of affected segments. In such situations, the RSU can choose an appropriate value of $r$, possibly after negotiation with the vehicles. For example, if a segment spans 250m but the jammer affects just 50m inside a segment, precise localization becomes challenging. To address this, the RSU may renegotiate finer segment granularity with vehicles, provided vehicles are willing to share location at a finer resolution. Input regarding the radiation capability of the jammer can be used by the RSU to appropriately select the number of segments $r$ into which the RSU divides its coverage area, and then this can be used to design the spatial-provenance framework accordingly to choose the appropriate parameters for localization, such as $m_1$, $m_2$, $k_1$, and $k_2$.\looseness=-1

\section{Discussion}\label{sec:discussion}

In this study, we adopt a linear segmentation model, assuming the coverage area of the RSU as a straight line, which is common in real-world deployments like highways and urban corridors \cite{1,3,4}. However, our proposed framework can be extended to other practical scenarios where the coverage area of the RSU can be segmented radially using concentric circles and angular sectors. In such a case, our proposed framework can be applied in each of its stages, including the modeling of relaxed-privacy, the embedding procedure at the vehicles, the recovery algorithm at the RSU, and the theoretical analysis on the optimization of embedding parameters with only minor modifications.\looseness=-1

The objective of this study was to collectively address the localization requirements of the RSU and the privacy concerns of the vehicles. Through the notion of relaxed-privacy, we have presented a low-latency method to embed and recover spatial-provenance. When using the proposed method, we have derived the analytical expressions on the false positive rate for learning the spatial-provenance. Extensive simulation results have been provided to validate the same. In addition, testbed results confirm that a few bits in the packet header for conveying the spatial-provenance will help the RSU to localize certain threats.\looseness=-1

In our method, we have used two correlated Bloom filters to convey the path and location information. As a result, we believe that false positive events arising from one Bloom filter can be potentially corrected by using the events from the other Bloom filter since they are correlated. Therefore, an intriguing direction for future research involves jointly optimizing the parameters of two Bloom filters $\textbf{BF}_1$ and $\textbf{BF}_2$. This work primarily focused on the homogeneous privacy of the vehicles, i.e., when all the vehicles agree to share the same precision of localization. In practice, vehicles may have varying privacy preferences. Some may agree to reveal location at a fine granularity (e.g., $3r$), while others prefer coarser levels (e.g., $2r$ or $r$). In such cases, to apply our homogeneous model, consensus can be based on the coarsest acceptable granularity among the vehicles. For instance, since agreeing to $3r$ inherently permits $2r$ or $r$, privacy is preserved unless finer granularity (e.g., $4r$) is enforced, which our model prevents. Therefore, our suggested protocol respects individual privacy thresholds even with heterogeneous privacy constraints.\looseness=-1 

\section*{Acknowledgement}
This work was supported by the project titled ``Development of Network Provenance Techniques for Monitoring Wireless Networks" under Contract for Acquisition of Research Services (CARS) from the Scientific Analysis Group (SAG), DRDO, New Delhi, India.\looseness=-1

% \printbibliography
\bibliographystyle{IEEEtran} % We choose the "plain" reference style

\bibliography{refs} 

\vspace{0.6cm}

\noindent \textbf{Manish Bansal} is a Ph.D Research Scholar at IIT Delhi. His research interests include 5G wireless security, provenance, sensor networks and IoT.

\vspace{0.1cm}

\noindent \textbf{Pramsu Shrivastava} holds a B.Tech in Electrical Engineering from IIT Delhi and has worked in wireless networks and computer architecture. He is passionate about low-level systems and exploring cutting-edge technologies.

\vspace{0.1cm}

\noindent \textbf{Harshan Jagadeesh} is an Associate Professor at the Department of Electrical Engineering, IIT Delhi. His research interests are in the broad areas of network security, information theory and coding theory.
\end{document}